\documentclass[11pt,letter]{article}
\usepackage{natbib}
\usepackage{latexsym}
\usepackage{graphics}
\usepackage{amsfonts}
\usepackage{enumerate}
\usepackage{graphicx}
\usepackage{epsf}
\usepackage{epstopdf}
\usepackage{epsfig}
\usepackage{verbatim}
\usepackage{subfig}
\usepackage{caption}
\usepackage{footnote}
\usepackage{slashbox}
\usepackage{lscape}
\usepackage{hhline}
\usepackage{array}
\usepackage{graphicx}
\usepackage{setspace}
\usepackage{amssymb}
\usepackage{amsmath}
\usepackage{amsfonts}
\usepackage{times}
\usepackage{color}
\usepackage{amsthm}
\usepackage{indentfirst}
\usepackage{epsfig}
\usepackage{lscape,ctable,rotating}
\usepackage{enumitem}
\usepackage{pdfpages}

\setlist{nolistsep}
\usepackage[left=1in,top=1in,right=1in,foot=1in]{geometry}

\usepackage{amsmath}

\DeclareMathOperator*{\argmin}{arg\,min}

\newtheorem{definition}{Definition}[section]
\newtheorem{lemma}[definition]{Lemma}

\newtheorem{exam}[definition]{Example}
\newtheorem{proposition}{Proposition}[section]
\newtheorem{rem}{Remark}[section]

\newcommand{\R}{\mathbb{R}}

\renewcommand{\P}{\mathbb{P}}

\newcommand{\var}{\textrm{\rm var}}
\newcommand{\cov}{\textrm{\rm cov}}
\newcommand{\im}{ \textrm{\rm Im} }

\newcommand{\levy}{{L\'evy}}

\newcommand{\Tau}{{\mathcal{T}}}

\newcommand{\disteq}{\stackrel{\mathrm{d}}{=}}

\newcommand{\stdNTS}{\textup{\rm stdNTS}}
\newcommand{\tr}{{\texttt{T}}}
\newcommand{\diag}{\textrm{\rm diag}}

\newcommand{\VaR}{\textup{\rm VaR}}

\newcommand{\CoVaR}[4]{\textup{\rm CoVaR}_{{#1},{#3}}\left({{#2}|{#4}} \right)}
\newcommand{\CoCVaR}[4]{\textup{\rm CoCVaR}_{{#1},{#3}}\left({{#2}|{#4}} \right)}

\title{Portfolio Optimization with Relative Tail Risk}
\author{Young Shin Kim\footnote{College of Business, Stony Brook University, New York, USA (aaron.kim@stonybrook.edu).} }
\providecommand{\keywords}[1]{\textbf{\textit{Key words:}} #1}

\begin{document}
\maketitle

\begin{abstract}
This paper proposes analytic forms of portfolio CoVaR and CoCVaR on the normal tempered stable market model.
Since CoCVaR captures the relative risk of the portfolio with respect to a benchmark return,
we apply it to the relative portfolio optimization.
Moreover, we derive analytic forms for the marginal contribution to CoVaR and the marginal contribution to CoCVaR. 
We discuss the Monte-Carlo simulation method to calculate CoCVaR and the marginal contributions of CoVaR and CoCVaR.
As the empirical illustration, we show relative portfolio optimization with thirty stocks under the distress condition of the Dow Jones Industrial Average. Finally, we perform the risk budgeting method to reduce the CoVaR and CoCVaR of the portfolio based on the marginal contributions to CoVaR and CoCVaR.\\
\end{abstract}

\baselineskip=24pt

\doublespacing
\section{Introduction}
Harry Markowitz's mean-variance portfolio optimization technique (\cite{Markowitz:1952}) has made a remarkable impact on portfolio theory in finance. This innovative method has been extensively utilized successfully in portfolio selection, allocation, and risk management practices. The adoption of Markowitz's mean-variance model has resulted in significant improvements in investment outcomes and become a cornerstone of modern portfolio theory. Many investors have utilized Markowitz's method to determine an efficient portfolio weight vector. Still, constructing a portfolio that surpasses the market index (e.g., S\&P 500 index or DJIA) is not easy to be achieved. For that reason, some investors use relative portfolio optimization, that is, relaxing the assumption of the portfolio variance as the tracking error, which is a variance of the relative return. Here, the relative return is the difference between the portfolio return and a benchmark return, such as the market index return

This paper presents how to improve relative portfolio optimization based on two aspects.
First, we loosen the Gaussian assumption of Markowitz's model, which was empirically rejected in literature including \cite{Fama:1963}, \cite{Mandelbrot:1963a,Mandelbrot:1963b}, and \cite{Cont_Tankov:2004}. 
Non-Gaussian multivariate distributions have been introduced to capture stylized facts, such as fat-tails and asymmetric dependence, that are not accounted for by the Gaussian model.
This paper proposes the normal tempered stable (NTS) distribution as an alternative distribution to the Gaussian distribution.
The NTS is defined by the tempered stable subordinated Gaussian distribution (\cite{BarndorffNielsen1978}, \cite{BarndorffNielsenLevendorskii:2001} and \cite{BarndorffNielsenShephard:2001}). The NTS distribution has been popularly applied in finance by capturing the fat-tails of the asset return distribution and describing asymmetric dependence (See \cite{KimVolkmann:2013}). For example, it is applied to financial risk management in \cite{Anad_et_al:2017} and \cite{KurosakiKim:2018}, and portfolio management in \cite{EberleinMadan:2010}, \cite{Anad_et_al:2016}, and \cite{kim2022portfolio}.
In addition, \cite{KIM2015512} applies a \levy~ process generated by the two-dimensional NTS distribution to Quanto option pricing, and its extension to capture stochastic dependence is further discussed in \cite{kim2023multi}.
Recently, the NTS distribution was applied to cryptocurrency portfolio optimization in \cite{kurosaki2022cryptocurrency}.

Second, we enhance the portfolio theory by replacing CoVaR and CoCVaR instead of the variance of the relative return as the risk measure.
Traditionally, portfolio variance, value at risk (VaR), and conditional value at risk (CVaR) have been employed as the portfolio risk measure, 
but those are absolute risk measures. Relative traders prefer to use tracking errors such as the variance of relative return with respect to the benchmark return instead of the absolute risk measures. In this paper, we take CoVaR (\cite{10.1257/aer.20120555}) and CoCVaR (\cite{HuangUryasev2018}) as the risk measures of relative portfolio optimization. CoVaR is proposed as a measure of systemic risk, which is the VaR of the financial system under the condition of a distressed market. \cite{GIRARDI20133169} calculate CoVaR on the multivariate GARCH model and \cite{REBOREDO2015214} use the copula method. CoCVaR (or CoES by \cite{10.1257/aer.20120555}) is an expected downturn of the financial system under the condition of a distressed market. \cite{HuangUryasev2018} define the mathematical formula of CoCVaR and apply it to measure the risk of the ten largest publicly traded banks in the United States under a distressed condition of market factors\footnote{VIX, liquidity spread, three-month treasury change, term spread change, credit spread change, equity market return, and real estate sector excess return.}. Recently, \cite{liu2021systemic} discuss CoVaR and CoCVaR on the GARCH model with two-dimensional NTS innovations and do the back-testing. 

The paper discusses portfolio optimization minimizing the portfolio CoVaR and CoCVaR concerning the market index and derives the marginal contributions to risk for the portfolio CoVaR and CoCVaR. The marginal contribution to risk is the rate of change in risk with respect to a small percentage change in proportion to a member's asset. Mathematically it is defined by the first derivative of CoVaR or CoCVaR with respect to the marginal weight. 
In the context of absolute optimization, the marginal contributions to VaR and CVaR are employed to identify assets with high and low levels of risk.
The general form of marginal risk contributions for the VaR and CVaR are provided in \cite{GourierouxaEtAl:2000}. The analytic forms of the marginal contributions for VaR and CVaR are discussed under the skewed-$t$ distribution in \cite{StoyanovRachevFabozzi:2013}, and under the NTS distribution in \cite{Kim_et_al:2012} and \cite{kim2022portfolio}.
Since we focus on relative optimization in this paper, we provide analytic formulas for the marginal contributions to CoVaR and CoCVaR under the NTS market model instead of VaR and CVaR. Those marginal contributions to CoVaR and CoCVaR help relative traders to make decisions in portfolio rebalancing. However, the multiple integrals involved in these formulas can present technical difficulties for numerical calculations. To solve this problem, we use Monte-Carlo simulation (MCS) methods. Furthermore, we perform the risk budgeting based on the marginal contributions to CoVaR and CVaR empirically.

The remainder of this paper is organized as follows. The definition of portfolio CoVaR and portfolio CoCVaR are presented in Section 2.
Section 3 reviews the NTS market model. Section 4 presents the portfolio CoVaR and CoCVaR on the NTS market model, along with a detailed analysis of the marginal contributions to CoVaR and CoCVaR. Empirical illustrations are given in section 5. In this section, we exhibit MCS method to calculate CoCVaR on the NTS market model and do the portfolio CoCVaR minimizing portfolio optimization with the estimated parameters. We also discuss portfolio budgeting using the marginal CoVaR and the marginal CoCVaR. Finally, Section 6 concludes. Proofs and mathematical details are presented in Appendix.

\section{CoVaR and CoCVaR}
Let $X$ be a random variable for a market factor, for instance, the market index returns, and $Y$ be a random variable for asset or  portfolio returns. We consider a random vector $(X, Y)$ following a joint distribution.
The CoVaR and CoCVaR of $Y$ at a significant level $\eta$ under a condition of the event $X\le -\VaR_\zeta(X)$ are defined by
\[
\CoVaR{\eta}{Y}{\zeta}{X}=\VaR_\eta(Y|X\le -\VaR_\zeta(X)) = -\inf_x \{x|P(Y\le x|X\le -\VaR_\zeta(X))\ge \eta \},
\]
and
\[
\CoCVaR{\eta}{Y}{\zeta}{X} = -E\left[Y|Y<-\CoVaR{\eta}{Y}{\zeta}{X}, X\le -\VaR_\zeta(X)\right],
\]
respectively\footnote{See \cite{10.1257/aer.20120555},\cite{HuangUryasev2018}, and \cite{liu2021systemic} for more details}, where $\VaR_\zeta(X)$ is the Value-at-Risk ($\VaR$) of $X$ for a significant level $\zeta$ given as
\[
\VaR_\zeta(X) = -\inf_x\{ x | P\left(X\le x\right)\ge\zeta\}.
\]
If the joint distribution of $(X,Y)$ is continuous, then $\VaR_\zeta(X) = -F_X^{-1}(\zeta)$ where 
$F_X$ is the cumulative distribution function (cdf) for the marginal distribution of $X$ and
$F_X^{-1}$ is the inverse function of $F_X$. Moreover, we have
\[
F_{(X,Y)}(-\VaR_\zeta(X), -\CoVaR{\eta}{Y}{\zeta}{X}) = \eta \zeta
\]
for the joint cdf of $(X,Y)$.

Consider $N$ number of stocks and a market index
and an $(N+1)$-dimensional random vector $R=(R_0$, $R_1$, $\cdots$, $R_N)^{\tr}$, where 
$R_0$ is the index return and $R_n$ is the return of the $n$th stock for $n\in\{1$, $2$, $\cdots$, $N\}$.
Let $w=(w_1$, $w_2$, $\cdots$, $w_N)^\tr\in I^N$ be an $N$-dimensional vector satisfying $\sum_{n=1}^N w_n= 1$ for $I=[0,1]$. The $n$th element $w_n$ of $w$ is the proportion of capital invested in the $n$th stock for $n\in\{1,2,\cdots, N\}$. In this case, $w$ be referred to as the \textit{capital allocation weight vector} for the $N$ stocks in the market\footnote{In this paper, we consider the long-only portfolio.}. Then we have the random variable of the portfolio return as $R_p(w) = \sum_{n=1}^N w_n, R_n$.

Suppose that $R$ follows $(N+1)$-dimensional normal distribution with a mean vector $\mu=(\mu_0$, $\mu_1$, $\cdots$, $\mu_N)^\tr$ and a covariance matrix $\Sigma$, that is $R\sim \Phi(\mu, \Sigma)$. Let $\sigma_{n,m}$ be the $(n+1, m+1)$-th element of $\Sigma$ for $n,m\in\{0,1,\cdots, N\}$. Then the bivariate random vector $(R_0, R_p(w))^\tr$  follows the bivariated normal distribution, $(R_0, R_p(w))^\tr\sim\Phi(\bar\mu, \bar\Sigma)$ with 
\[
\bar\mu = (\mu_0, \mu_p(w))^\tr,  \text{ and }
\bar\Sigma = \left(\begin{matrix}
		\sigma_{0,0}& \sigma_{p,0}(w)\\
		 \sigma_{p,0}(w)& \sigma_{p}^2(w)
	\end{matrix}\right)
,
\]
where
$\mu_p(w) = \sum_{n=1}^N w_n \mu_n$, $\sigma_{p,0}(w)=\sum_{n-1}^Nw_m\sigma_{0,n}$, $\sigma_p^2(w) = \sum_{n=1}^N\sum_{m=1}^N w_m w_m\sigma_{n,m}$.
We can calculate VaR of $R_0$ as $\VaR_\zeta(R_0) = q_{1-\zeta}\sqrt{\sigma_{0,0}}-\mu_0$, where $ q_{1-\zeta}$ is $(1-\zeta)$-quantile value of the standard normal distribution. Suppose the value $x$ satisfies
\[
\int_{-\infty}^x\int_{-\infty}^{-\VaR_\zeta(R_0)} f_{\Phi(\bar\mu, \bar\Sigma)}(x_1,x_2) dx_1 dx_2 = \eta\zeta,
\]
where $f_{\Phi(\bar\mu, \bar\Sigma)}$ is the probability density function of $\Phi(\bar\mu, \bar\Sigma)$.
Then $\CoVaR{\eta}{R_p(w)}{\zeta}{R_0}=-x$. Moreover,
\[
\CoCVaR{\eta}{R_p(w)}{\zeta}{R_0} = -\frac{1}{\eta\zeta}\int_{-\infty}^{-\CoVaR{\eta}{R_p(w)}{\zeta}{R_0}}\int_{-\infty}^{-\VaR_\zeta(R_0)} x_1 f_{\Phi(\bar\mu, \bar\Sigma)}(x_1,x_2) dx_1 dx_2.
\]

\section{Standard NTS distribution and NTS Market model}
In this section, we define the multivariate standard NTS distribution and construct the market model based on the distribution. 
\subsection{Standard NTS Distribution}
Let $N$ be a finite positive integer and $\Xi=(\Xi_1, \Xi_2, \cdots, \Xi_N)^{\tr}$ be a multivariate random variable given by
\[
\Xi = \beta(\Tau-1) + \textup{diag}(\gamma) \varepsilon \sqrt{\Tau} ,
\]
where 
\begin{itemize}
\item $\Tau$ is the tempered stable subordinator with parameters $(\alpha,\theta)$, and is independent of $\varepsilon_n$ for all $n=1,2,\cdots, N$.
\item $\beta = (\beta_1, \beta_2, \cdots, \beta_N)^{\tr}\in\R^N$ with $|\beta_n|<\sqrt\frac{2\theta}{2-\alpha}$ for all $n\in\{1,2,\cdots, N\}$.
\item $\gamma = (\gamma_1, \gamma_2, \cdots, \gamma_N)^{\tr}\in\R_+^N$ with $\gamma_n = \sqrt{1-\beta_n^2\left(\frac{2-\alpha}{2\theta}\right) }$  for all $n\in\{1,2,\cdots, N\}$ and $\R_+=[0,\infty)$.
\item $\varepsilon = (\varepsilon_1, \varepsilon_2, \cdots, \varepsilon_N)^{\tr}$ is $N$-dimensional standard normal distribution with a covariance matrix $\Sigma$. That is, $\varepsilon_n\sim \Phi(0,1)$ for $n\in\{1,2,\cdots, N\}$ and $(k,l)$-th element of $\Sigma$ is given by $\rho_{k,l}=\cov(\varepsilon_k,\varepsilon_l)$ for $k,l\in\{1,2,\cdots,N\}$. Note that $\rho_{k,k}=1$.
\end{itemize}
In this case, $\Xi$ is referred to as the \textit{$N$-dimensional standard NTS random variable} with parameters $(\alpha$, $\theta$, $\beta$, $\Sigma)$ and we denote it by $\Xi\sim \textup{stdNTS}_N(\alpha$, $\theta$, $\beta$, $\Sigma)$ (See more details in \cite{KimVolkmann:2013} and \cite{kim2022portfolio}.).
  \begin{itemize}
  \item The probability density function (pdf) $f_\Tau(t)$ of $\Tau$ is $f_\Tau(t) = \frac{1}{2\pi}\int_{-\infty}^\infty e^{-iut}\phi_\Tau(u)du$, where $\phi_{\Tau}$ is the characteristic function (ch.f) of $\Tau$ given by 
$
\phi_{\Tau}(u) =
\exp\left(-\frac{2\theta^{1-\frac{\alpha}{2}}}{\alpha}\left((\theta-iu)^{\frac{\alpha}{2}}-\theta^{\frac{\alpha}{2}}\right)\right).
$
  \item The cdf of the stdNTS vector $\Xi$ is given by 
  \begin{equation}\label{eq:cdf nts}
  F_\Xi(a_1, \cdots, a_N) = \int_0^\infty \int_{-\infty}^{\frac{a_N-\beta_N(t-1)}{\gamma_N\sqrt{t}}}\cdots\int_{-\infty}^{\frac{a_1-\beta_1(t-1)}{\gamma_1\sqrt{t}}}f_\varepsilon(x_1, \cdots, x_N)dx_1\cdots dx_N f_\Tau(t)dt,
  \end{equation}
where
 $f_\varepsilon$ is the pdf of $N$-dimensional normal distribution with mean $0$ and covariance $\Sigma$.
  \item The pdf if the vector $\Xi$ is given by 
  \begin{equation}\label{eq:pdf nts}
  f_\Xi(x_1, \cdots, x_N) = \int_0^\infty f_{\Phi(m(t), \Sigma(t))}(x_1, \cdots, x_N)f_\Tau(t) dt
  \end{equation}
  where $f_{\Phi(m(t), \Sigma(t))}(x_1, \cdots, x_N)$ is the pdf of the $N$-dimensional normal distribution with mean $m(t)=(\beta_1(t-1), \cdots, \beta_N(t-1))$ and covariance $\Sigma(t)=\left(t\gamma_k \gamma_l \rho_{k,l}\right)_{k,l\in\{1,\cdots, N\}}$.
\item By \cite{Gil-Pelaez:1951}, the marginal cdf of $\Xi_n$ for $n\in\{1,2,\cdots, N\}$ is equal to 
\begin{equation}\label{eq:cdfStdNTS}
F_{\Xi_n}(x) = \frac{1}{2}-\frac{1}{\pi}\int_{-\infty}^\infty \frac{\im (e^{-iux} \phi_{\Xi_n}(u))}{u} du.
\end{equation}
where $\phi_{\Xi_n}$ is the ch.f of $\Xi_n$ given by
\begin{align*}
\phi_{\Xi_n}(u)&=\exp\left(-\beta_niu-\frac{2\theta^{1-\frac{\alpha}{2}}}{\alpha}
\left(\left(\theta-i\beta_n u+\frac{u^2}{2}\left(
1-\beta_n^2\left(\frac{2-\alpha}{2\theta}\right)
\right)\right)^{\frac{\alpha}{2}}-\theta^{\frac{\alpha}{2}}\right)\right).
\end{align*}
Moreover, the marginal pdf of $\Xi_n$ is obtained by the inverse Fourier transform for the characteristic function as
\[
f_{\Xi_n}(x) = \frac{1}{2\pi}\int_{-\infty}^\infty e^{-iux} \phi_{\Xi_n}(u) du.
\]
  \item Covariance between $\Xi_n$ and $\Xi_m$ for $n,m\in\{1,2,\cdots, N\}$ is equal to
  \begin{equation}\label{eq:cov Xi_n and Xi_m}
  \cov(\Xi_n, \Xi_m)=\gamma_n\gamma_m\rho_{n,m}+\beta_n\beta_m\left(\frac{2-\alpha}{2\theta}\right).
  \end{equation}
  \end{itemize}
\subsection{NTS Market Model}
Consider a portfolio having $N$ assets for a positive integer $N$. The return of the assets in the portfolio is given by a random vector $R=(R_1$, $R_2$, $\cdots$, $R_N)^\tr$. We suppose that the return $R$ follows 
\begin{equation} \label{eq:NTS Market}
R = \mu + \diag(\sigma) \Xi 
\end{equation}
where  $\mu = (\mu_1, \mu_2, \cdots, \mu_N)^\tr\in\R^N$, $\sigma = (\sigma_1, \sigma_2, \cdots, \sigma_N)^\tr\in\R_+^N$ and $\Xi\sim \stdNTS_N(\alpha, \theta, \beta, \Sigma)$. Then, we have $E[R_n]=\mu_n$ and $\var(R_n) = \sigma_n^2$ for all $n\in\{1,2, \cdots, N\}$. This market model is referred to as the \textit{NTS market model}.

Let $w=(w_1, w_2, \cdots, w_N)^\tr$ be the capital allocation weight vector. Then the portfolio return for $w$ is equal to $R_P(w)=w^\tr R$. The distribution of $R_P(w)$ is presented in the following proposition from \cite{kim2022portfolio}.
\begin{proposition}[\cite{kim2022portfolio}]\label{pro:mu+sigma stdNTS} 
Consider a portfolio with $N$ assets for a positive integer $N$.
Let $w = (w_1$, $w_2$, $\cdots$, $w_N)^{\tr}$ be the capital allocation weight vector for the portfolio.
Suppose $R=(R_1$, $R_2$, $\cdots$, $R_N)^\tr\in\R^N$ follows the NTS market model given by \eqref{eq:NTS Market} with $\mu \in\R^N$,
$\sigma \in \R_+^N$, and $\Xi \sim \textup{stdNTS}_N(\alpha, \theta, \beta, \Sigma)$.  Then
\begin{equation} \label{eq:NTS Portfolio Return}
R_P(w) \disteq \bar\mu(w) + \bar\sigma(w) \Xi ~~~\text{ for }~~~ \Xi\sim \textup{stdNTS}_1(\alpha, \theta, \bar\beta(w), 1),
\end{equation}
where 
\[
\bar\mu(w) = \sum_{n=1}^N w_n \mu_n,~~~ \bar\sigma(w) = \sqrt{\sum_{n=1}^N w_n \sum_{m=1}^N w_n w_m \cov(R_n, R_m)},
~~~\text{and}~~~\bar\beta(w)=\frac{\sum_{n=1}^N  w_n\sigma_n\beta_n}{\sigma_p(w)}.
\]
\end{proposition}

\section{\label{sec:NTS}Portfolio CoVaR and CoCVaR with respect to a Market Index}
Consider $N$ number of stocks and a market index for a given market as follows:
\begin{itemize}
\item $R_0$ is the market index return.
\item $R_n$ for $n\in\{1,2,\cdots, N\}$ is the individual stock return of a stock portfolio consisting of $N$ stocks.
\end{itemize}
Suppose $R=(R_0, R_1, \cdots, R_N)^{\tr}$ follows a ($N+1$)-dimensional NTS market model, that is, we have
\begin{equation}\label{eq:nts market model with index}
R =\mu+\textup{diag}(\sigma)\Xi \text{ with } \Xi\sim \textup{stdNTS}_{N+1}(\alpha, \theta, \beta, \Sigma)
\end{equation}
where
\begin{itemize}
\item $\mu=(\mu_0, \mu_1, \cdots, \mu_N)^{\tr}\in\R^{N+1}$, and $\sigma=(\sigma_0, \sigma_1, \cdots, \sigma_N)^{\tr}\in\R_+^{N+1}$,  
\item $\alpha\in(0,2)$, $\theta>0$, $\beta=(\beta_0, \beta_1, \cdots, \beta_N)^{\tr}\in\R^{N+1}$ with $|\beta_n|<\sqrt{\frac{2\theta}{2-\alpha}}$ for $n\in\{0,1,\cdots, N\}$,
\item $\Sigma$ is the $(N+1)\times(N+1)$ covariance matrix and $\rho_{n,m}$ is the $(n+1,m+1)$-th element of $\Sigma$ for $n,m\in\{0,1,\cdots, N\}$.
\end{itemize} 
Based on the assumption, we obtain the following proposition whose proof is in Appendix.
\begin{proposition}\label{prop:Index and Portfolio}
Suppose $R=(R_0, R_1, \cdots, R_N)^{\tr}$ follows a ($N+1$)-dimensional NTS market model as \eqref{eq:nts market model with index}.
Let $w = (w_1, w_2, \cdots, w_N)^{\tr}$ be a capital allocation weight vector and $R_p(w)=\sum_{n=1}^N w_n R_n$.
The bivariate random vector $(R_0, R_P(w))^\tr$ follows 2-dimensional NTS model such as 
\[
\left(\begin{matrix}
R_0 \\ R_p(w)
\end{matrix}\right)
=
\left(\begin{matrix}
\mu_0\\ \mu_p(w)
\end{matrix}\right)
+
\left(\begin{matrix}
\sigma_0 & 0\\ 0 &\sigma_p(w)
\end{matrix}\right)
\left(\begin{matrix}
\Xi_0 \\ \Xi_p(w)
\end{matrix}\right)
\]
for
\[
\left(\begin{matrix}
\Xi_0 \\ \Xi_p(w)
\end{matrix}\right)\sim \textup{stdNTS}_2\left(\alpha, \theta, 
\left(\begin{matrix}
\beta_0 \\ \beta_p(w)
\end{matrix}\right),
\left(\begin{matrix}
1 & \rho_p(w)\\ \rho_p(w) & 1
\end{matrix}\right)
\right),
\]
where
\[
\mu_p(w) = \sum_{n=1}^N w_n \mu_n,~~~ 
\sigma_p(w) 
= \sqrt{\sum_{n=1}^N\sum_{m=1}^N w_n w_m \sigma_n \sigma_m \cov\left(\Xi_n, \Xi_m\right)},~~~
\beta_p(w) = \frac{\sum_{n=1}^N  w_n\sigma_n\beta_n}{\sigma_p(w)},
\]
and
\begin{align}\label{eq:rho 0,p}
\rho_p(w)&:=\frac{\sum_{n=1}^N  w_n\gamma_n\sigma_n \rho_{0,n}}{\sqrt{\sum_{n=1}^N \sum_{m=1}^N w_n w_m\gamma_n\gamma_m\sigma_n\sigma_m\rho_{n,m}}} 
\end{align}
for $\gamma_n = \sqrt{1-\beta_n^2\left(\frac{2-\alpha}{2\theta}\right)}$.
\end{proposition}
\textbf{Remark}
Let $\sigma_* = (\sigma_1, \sigma_2, \cdots, \sigma_N)^\tr$, $\rho^*=(\rho_{0,1}, \rho_{0,2}, \cdots, \rho_{0,N})^\tr$, $\Sigma^* = \left(\rho_{n,m}\right)_{n,m\in\{1,2,\cdots, N\}}$ and $\gamma^* = (\gamma_1, \gamma_2, \cdots, \gamma_N)^\tr$. Then
\begin{align*}
\rho_p(w)&:=\frac{w^\tr V_{\gamma,\sigma, \rho}^*}{\sqrt{w^\tr \Sigma^*_{\gamma, \sigma} w} }
\end{align*}
for 
$V_{\gamma,\sigma, \rho}^* = \diag(\gamma^*)\diag(\sigma^*)\rho^*$ 
and 
$\Sigma_{\gamma, \sigma}^* =  \diag(\gamma^*)\diag(\sigma^*)\Sigma^*\diag(\sigma^*)\diag(\gamma^*)$.

With the positive homogeneity and translation invariance properties of VaR and CoVaR, we have
\[
\VaR_\zeta(R_0) = \sigma_0(w) \VaR_\zeta(\Xi_0) - \mu_0
\]
and 
\[
\CoVaR{\eta}{R_p(w)}{\zeta}{R_0} = \sigma_p(w)\CoVaR{\eta}{\Xi_p(w)}{\zeta}{\Xi_0}-\mu_p(w).
\]  
Since marginal distributions of $\Xi_0$, and $\Xi_p(w)$ are continuous, we have $\VaR_\zeta(\Xi_0) = -F_{\Xi_0}^{-1}(\zeta)$, where $F_{\Xi_0}$ and $F_{\Xi_0}^{-1}$ are the cdf and inverse cdf of $\Xi_0$, respectively.
Moreover, $\CoVaR{\eta}{\Xi_p(w)}{\zeta}{\Xi_0}$ is the value satisfying
\[
\P\left(\Xi_p(w)<-\CoVaR{\eta}{\Xi_p(w)}{\zeta}{\Xi_0}, \Xi_0<-\VaR_\zeta(\Xi_0)\right) = \zeta\eta,
\]
i.e.
\begin{equation}\label{eq:CoVaR Implicit}
F_{(\Xi_0,\Xi_p(w))}\left(-\VaR_\zeta(\Xi_0), -\CoVaR{\eta}{\Xi_p(w)}{\zeta}{\Xi_0}\right) = \zeta\eta,
\end{equation}
where $F_{(\Xi_0,\Xi_p(w))}$ is the cdf of $(\Xi_0,\Xi_p(w))$.
By \eqref{eq:cdf nts}, $F_{(\Xi_0,\Xi_p(w))}$ is given as
\begin{align*}
&F_{(\Xi_0,\Xi_p(w))}(\xi_0, \xi_p) \\
&= \int_0^\infty \int_{-\infty}^{\frac{ \xi_p-\beta_p(w)(t-1)}{\gamma_p(w)\sqrt{t}}} \int_{-\infty}^{\frac{\xi_0-\beta_0(t-1)}{\gamma_0\sqrt{t}}} f^{\Phi_2}_{\rho_p(w)}(x_1, x_2) dx_1 dx_2f_\Tau(t)dt,
\end{align*}
where $\gamma_0= \sqrt{1-\beta_0^2\left(\frac{2\theta}{2-\alpha}\right)}$, $\gamma_p(w)= \sqrt{1-(\beta_p(w))^2\left(\frac{2\theta}{2-\alpha}\right)}$ and $f^{\Phi_2}_\rho$ is the pdf of the bivariate standard normal distribution with covariance $\rho$.

By the same arguments, we have CoCVaR as
\begin{equation}
\label{eq:CoCVaR Portfolio}
\CoCVaR{\eta}{R_p(w)}{\zeta}{R_0} = \sigma_p(w) \CoCVaR{\eta}{\Xi_p(w)}{\zeta}{\Xi_0} - \mu_p(w)
\end{equation}
and
\begin{align*}
\CoCVaR{\eta}{\Xi_p(w)}{\zeta}{\Xi_0}
  =-\frac{1}{\zeta\eta}\int_{-\infty}^{-\CoVaR{\eta}{\Xi_p(w)}{\zeta}{\Xi_0}} \int_{-\infty}^{-\VaR_\zeta(\Xi_0)} x_2 f_{(\Xi_0, \Xi_p(w))}(x_1,x_2) dx_1\,dx_2
\end{align*}
where
\begin{align*}
  &f_{(\Xi_0, \Xi_p(w))}(x_1, x_2) 
  =   \int_0^\infty 
g\left(x_1-\beta_0(t-1),x_2-\beta_p(w)(t-1), \gamma_0\sqrt{t}, \gamma_p(w)\sqrt{t}, \rho_p(w)\right)  f_\Tau(t) dt,
\end{align*}
and
\[
g(x, y, u, v, \rho)=\frac{1}{2\pi u v \sqrt{1-\rho}}  
\exp\left(
	-\frac{1}{2(1-\rho^2)}
	\left(\left(\frac{x}{u}\right)^2-\frac{2\rho xy}{uv}+\left(\frac{y}{v}\right)^2\right)
\right),
\]
by \eqref{eq:pdf nts}. By the change of variables, we have
\begin{align}
\nonumber
&\CoCVaR{\eta}{\Xi_p(w)}{\zeta}{\Xi_0}
\\
\label{eq:CoCVaR Integral Form}
&
  =-\frac{1}{\zeta\eta}\int_0^\infty \int_{-\infty}^{C(w,t)} \int_{-\infty}^{v(t)}
   \left(\beta_p(w)(t-1)+x_2\gamma_p(w)\sqrt{t}\right)f^{\Phi_2}_{\rho_p(w)}(x_1, x_2) dx_1\,dx_2 f_\Tau(t)dt,
\end{align}
or  
\begin{align}
\label{eq:CoCVaR Expectation Form}
&\CoCVaR{\eta}{\Xi_p(w)}{\zeta}{\Xi_0}
  =-\frac{1}{\zeta\eta}E\left[\left(\beta_p(w)(\Tau-1)+\epsilon_p\gamma_p(w)\sqrt{\Tau}\right)1_{\epsilon_p<C(w,\Tau)}1_{\epsilon_0<v(\Tau)}\right]
\end{align}
where 
$C(w,t) = \frac{ -\CoVaR{\eta}{\Xi_p(w)}{\zeta}{\Xi_0}-\beta_p(w)(t-1)}{\gamma_p(w)\sqrt{t}}$, 
$v(t) = \frac{\xi_0-\beta_0(t-1)}{\gamma_0\sqrt{t}}$,
 $\xi_0 = -\VaR_\zeta(\Xi_0)$, $(\epsilon_0, \epsilon_p)$ is the bivariate standard normal random vector with covariance $\rho_p(w)$, and $\Tau$ is the tempered stable subordinator with parameters $(\alpha, \theta)$ independent of $(\epsilon_0, \epsilon_p)$.

\subsection{\label{Sec:MCTCoCVaR}Marginal Contribution to CoVaR and CoCVaR}
In this section, we discuss the marginal contributions to CoVaR and CoCVaR under the $(N+1)$-dimensional NTS market model with the market index defined in the previous section. Let $w_j$ be the $j$th element of the capital allocation vector $w = (w_1$, $\cdots$, $w_N)^\tr\in I^N$ for $j\in\{1$, $2$, $\cdots$, $N\}$.
Let
$u(x,w,t) = \frac{ x-\beta_p(w)(t-1)}{\gamma_p(w)\sqrt{t}}$, 
and $v(t) = \frac{v_0-\beta_0(t-1)}{\gamma_0\sqrt{t}}$,
where $v_0 = -\VaR_\eta(\Xi_0)$, 
$\gamma_p(w) = \sqrt{1-(\beta_p(w))^2\left(\frac{2\theta}{2-\alpha}\right)}$,
and $\gamma_0 = \sqrt{1-\beta_0^2\left(\frac{2\theta}{2-\alpha}\right)}$.
Let $f^{\Phi_2}_{\rho_p(w)}$be the bivariate standard normal pdf with covariance $\rho_p(w)$, that is
\[
f^{\Phi_2}_{\rho_p(w)}(x_1,x_2)=
\frac{1}{2\pi\sqrt{1-(\rho_p(w))^2}}\exp\left(-\frac{x_1^2-2\rho_p(w) x_1 x_2+x_2^2}{2(1-(\rho_p(w))^2)}\right).
\]
Define
\begin{align*}
&G(x, w) 
= \int_0^\infty 
\int_{-\infty}^{u(x,w,t)} 
\int_{-\infty}^{v(t)} f^{\Phi_2}_{\rho_p(w)}(x_1,x_2)dx_1 dx_2  f_\Tau(t)dt - \eta\zeta.
\end{align*}
Then
$G(x, w) = 0$
if $x = -\CoVaR{\eta}{\Xi_p(w)}{\zeta}{\Xi_0}$ by \eqref{eq:CoVaR Implicit}.

\begin{proposition}\label{prop:mctcovar}
The marginal contribution to CoVaR for the $j$th element of the portfolio capital allocation weight vector is
\begin{align}\label{eq:mctcovar}
\frac{\partial}{\partial w_j}\CoVaR{\eta}{\Xi_p(w)}{\zeta}{\Xi_0} = \frac{\frac{\partial }{\partial w_j} G(x,w) }{\frac{\partial }{\partial x} G(x,w)}\Bigg|_{x = -\CoVaR{\eta}{\Xi_p(w)}{\zeta}{\Xi_0}}
\end{align}
with
\begin{align}\label{eq:dGdx}
&\frac{\partial}{\partial x} G(x,w)
=\frac{1}{\sqrt{2\pi}\gamma_p(w)}E\left[\frac{1}{\sqrt{\Tau}}\exp\left(-\frac{\left(u(x,w,\Tau)\right)^2}{2}\right) 
F_{\Phi}\left(\frac{v(\Tau) -\rho_p(w)u(x,w,\Tau)} {\sqrt{1-(\rho_p(w))^2}}\right)  \right]
\end{align}
and
\begin{align}
&\frac{\partial }{\partial w_j} G(x,w) \label{eq:dGdwj} 
 =\frac{\rho_p(w)}{1-(\rho_p(w))^2} \left(\frac{\partial}{\partial w_j}\rho_p(w) \right)F_{(\Xi_0, \Xi_p(w))}\left(v_0, x\right)
 \\
\nonumber  
&+ \frac{\frac{\partial}{\partial w_j}\rho_p(w)}{(1-(\rho_p(w))^2)^2} E\left[(\rho_p(w)\epsilon_p^2-(1+(\rho_p(w))^2)\epsilon_p\epsilon_0+\rho_p(w)
\epsilon_0^2)1_{\epsilon_p<u(x,w,\Tau)}1_{\epsilon_0<v(\Tau)}\right]
\\
\nonumber 
&+\frac{1}{\sqrt{2\pi}}E\left[\exp\left(-\frac{\left(u(x,w,\Tau)\right)^2}{2}\right) 
F_{\Phi}\left(\frac{v(\Tau) -\rho_p(w)u(x,w,\Tau)} {\sqrt{1-(\rho_p(w))^2}}\right)  
\frac{\partial}{\partial w_j}u(x,w,\Tau)\right]
\end{align}
where

\begin{itemize}
\item  $F_{\Phi}$ is the cdf of the standard normal distribution,
\item  $\frac{\partial}{\partial w_j}\rho_p(w)$ and $\frac{\partial}{\partial w_j}u(x,w,t)$ are given by \eqref{eq:d rho d wj} and \eqref{eq:d u d wj} in Appendix, respectively, 
\item $(\epsilon_0, \epsilon_p)$ is the bivariate standard normal distributed random vector with covariance $\rho_p(w)$, 
\item and $\Tau$ is the tempered stable subordinator with parameters $(\alpha, \theta)$ independent of $(\epsilon_0, \epsilon_p)$.
\end{itemize}
Here, $F_{(\Xi_0, \Xi_p(w))}\left(v_0, x\right)=\eta\zeta$ if $x = -\CoVaR{\eta}{\Xi_p(w)}{\zeta}{\Xi_0}$.
\end{proposition}

\begin{proposition}\label{prop:mctCoCVaR}
The marginal contribution to CoCVaR for the $j$-th element of the portfolio capital allocation weight vector is
\begin{align}
\label{eq:mctcocvar}
&\frac{\partial}{\partial w_j}\CoCVaR{\eta}{\Xi_p(w)}{\zeta}{\Xi_0} 
\\
\nonumber
&=
-\frac{1}{\eta\zeta}\Bigg(
E\left[\left((\Tau-1)-\frac{\epsilon_p\beta_p(w)\left(\frac{2\theta}{2-\alpha}\right)\sqrt{\Tau}}{\gamma_p(w)}\right) 1_{\epsilon_p<C(w,\Tau)}1_{\epsilon_0<v(\Tau)}\right]\frac{\partial}{\partial w_j} \beta_p(w)
\\
\nonumber
&
\hspace{1.5cm}
- E\Bigg[ \left(\frac{\rho_p(w)\epsilon_0^2-(1+(\rho_p(w))^2)\epsilon_0 \epsilon_p+\rho_p(w) \epsilon_p^2}{(1-(\rho_p(w))^2)^2}\right)
\\
\nonumber
&
\hspace{3cm}\times
(\xi_p(w)+\CoCVaR{\eta}{\Xi_p(w)}{\zeta}{\Xi_0}) 1_{\epsilon_p<C(w,\Tau)}1_{\epsilon_0<v(\Tau)}\Bigg]\frac{\partial}{\partial w_j} \rho_p(w)
\Bigg).
\end{align}
where 
\begin{itemize}
\item $C(w,t)=u(- \CoVaR{\eta}{\Xi_p(w)}{\zeta}{\Xi_0}, w, t)$,
\item $\xi_p(w) = \beta_p(w)(\Tau-1)+\epsilon_p \gamma_p(w)\sqrt{\Tau}$,
\item $\frac{\partial}{\partial w_j} \beta_p(w)$ and $\frac{\partial}{\partial w_j}\rho_p(w)$ are given by \eqref{eq:d beta d wj} and \eqref{eq:d rho d wj} in Appendix, respectively,
\item $(\epsilon_0, \epsilon_p)$ is the bivariate standard normal distributed random vector with covariance $\rho_p(w)$, 
\item and $\Tau$ is the tempered stable subordinator with parameters $(\alpha, \theta)$ independent of $(\epsilon_0, \epsilon_p)$.
\end{itemize}
\end{proposition}

We deduce the marginal contribution to CoVaR as follows:
\begin{align*}
&\frac{\partial}{\partial w_j}\CoVaR{\eta}{R_p(w)}{\zeta}{R_0}
\\
&=\CoVaR{\eta}{\Xi_p(w)}{\zeta}{\Xi_0}\frac{\partial}{\partial w_j}\sigma_p(w) 
+\sigma_p(w)\frac{\partial}{\partial w_j}\CoVaR{\eta}{\Xi_p(w)}{\zeta}{\Xi_0}
-\frac{\partial}{\partial w_j}\mu_p(w),
\end{align*} 
or
\begin{align}
\label{eq:mct CoVaR Portfolio}
&\frac{\partial}{\partial w_j}\CoVaR{\eta}{R_p(w)}{\zeta}{R_0}
\\
\nonumber
&=\frac{\sum_{n=1}^Nw_n\sigma_n\sigma_j\cov(\Xi_n,\Xi_j)}{\sigma_p(w)}\CoVaR{\eta}{\Xi_p(w)}{\zeta}{\Xi_0}
+\sigma_p(w)\frac{\partial}{\partial w_j}\CoVaR{\eta}{\Xi_p(w)}{\zeta}{\Xi_0}
-\mu_j
\end{align} 
where $\frac{\partial}{\partial w_j}\CoVaR{\eta}{\Xi_p(w)}{\zeta}{\Xi_0}$ is given as \eqref{eq:mctcovar} in Proposition \ref{prop:mctcovar}.
By the same arguments, we obtain the marginal contribution to CoCVaR as follows
\begin{align}
\label{eq:mct CoCVaR Portfolio}
&\frac{\partial}{\partial w_j}\CoCVaR{\eta}{R_p(w)}{\zeta}{R_0}
\\
\nonumber
&=\frac{\sum_{n=1}^Nw_n\sigma_n\sigma_j\cov(\Xi_n,\Xi_j)}{\sigma_p(w)}\CoCVaR{\eta}{\Xi_p(w)}{\zeta}{\Xi_0}
+\sigma_p(w)\frac{\partial}{\partial w_j}\CoCVaR{\eta}{\Xi_p(w)}{\zeta}{\Xi_0}
-\mu_j,
\end{align} 
where $\frac{\partial}{\partial w_j}\CoVaR{\eta}{\Xi_p(w)}{\zeta}{\Xi_0}$ is given as \eqref{eq:mctcocvar} in Proposition \ref{prop:mctCoCVaR}.

\begin{table}[t]
\centering
\begin{tabular}{cc|cc}
\hline
Company & Symbol & Company & Symbol \\
\hline
\hline
3M	&	 MMM	&
American Express	&	 AXP	\\
Amgen	& AMGN &
Apple Inc.	&	AAPL	 \\
Boeing	&	 BA	&
Caterpillar Inc.	&	 CAT	\\
Chevron Corporation	&	 CVX	&
Cisco Systems	&	CSCO	\\
The Coca-Cola Company	&	 KO	&
DuPont de Nemours Inc.	 &	 DD	\\
Goldman Sachs	&	 GS	&
The Home Depot	&	 HD	\\
Honeywell	& HON
IBM	&	 IBM	\\
Intel	&	INTC	&
Johnson \& Johnson	&	 JNJ	\\
JPMorgan Chase	&	 JPM	&
McDonald's	&	 MCD	\\
Merck \& Co.	&	 MRK	&
Microsoft	&	MSFT	\\
Nike	 &	 NKE	&
Procter \& Gamble	&	 PG	\\
Salesforce	& CRM &
The Travelers Companies	&	 TRV	\\
United Health Group	&	 UNH	&
Verizon	&	 VZ	\\
Visa Inc.	&	 V	&
Walmart	&	 WMT	\\
Walgreens Boots Alliance	&	WBA	&
The Walt Disney Company	&	 DIS	\\
\hline
\end{tabular}
\caption{\label{table:DJIA Members}Companies and symbols of 30 Stocks. They are selected based on the components for DJIA index as of 2022, but Dow Inc.(DOW) in the components is replaced by DuPont de Nemours Inc.(DD).}
\end{table}

\section{Empirical Illustration}
We fit the parameters of the NTS market model to Dow Johns Industrial Average (DJIA) index and 30 major stocks\footnote{The selected 30 stocks are the components of the DJIA index as of 2022. Since Dow Inc.(DOW) in the components does not have enough history, it is replaced by DuPont de Nemours Inc.(DD).} in the U.S. stock market. The company names and symbols of those 30 stocks are listed in Table \ref{table:DJIA Members}.

The parameter estimation method in this section is similar as the method in \cite{kim2022portfolio}. 
We take the set of daily log returns for the DJIA index and each stock from 11/27/2018 to 11/15/2022
and calculate sample means and sample standard deviations for each stock and the index.
The residuals are extracted by the z-score, and then fit the stdNTS parameters to the residuals of each stock (or index) return.
The curve fit method is used between the cdf of the stdNTS distribution obtained by \eqref{eq:cdfStdNTS} and the empirical cdf obtained by the kernel density estimation. 
The same as \cite{kim2022portfolio}, we use the index-based method with DJIA index data in order to find $\alpha$ and $\theta$ and then estimate the $\beta$ vector and $\Sigma$ matrix. That is, we find the 30-dimensional stdNTS parameters as the following two-step method:
\begin{itemize}
\item[] \textbf{Step 1} Find $(\alpha, \theta, \beta_0)$ using the curve fit method between the empirical cdf and stdNTS cdf for the residual of DJIA index. The parameters $(\alpha, \theta)$ are considered the parameters of the tempered stable subordinator.
\item[] \textbf{Step 2} Taking $(\alpha, \theta)$ estimated at Step 1, find $\beta_n$ by applying the curve fit method with the fixed $\alpha$ and $\theta$ for each $n$-th stock returns $n\in\{1,2,\cdots, 30\}$.
\item[] \textbf{Step 3} Find sample covariance matrix $(\cov(\Xi_n, \Xi_m))_{n,m\in\{0,1,\cdots, 30\}}$ using the standardized residual. Here the index $0$ is assigned to the DJIA index. Find $\rho_{n,m}$ using \eqref{eq:cov Xi_n and Xi_m}.
\end{itemize}
The parameters of the tempered stable subordinator in this investigation are $\alpha=1.1835$ and $\theta=0.0820$.
The other estimated parameters are presented in Table \ref{table:ParamEst} with Kolmogorov-Smirnov (KS) p-values present in the table for the goodness of fit test.

\subsection{\label{sec:mcs CoVaR CoCVaR}Calculating CoCVaR, MCT-CoVaR and MCT-CoCVaR with MCS}
In this section, we discuss MCS to find portfolio CoCVaR, MCT-CoVaR and MCT-CoCVaR using the parameters in Table \ref{table:ParamEst} on the NTS market model. 
Here, we use R language version 4.2.2 running on MS-Windows$^\text{\textcircled{R}}$ 10 operating system with the processor Intel$^\text{\textcircled{R}}$ Core i7-4790, 3.60GHz.
To calculate CVaR of the portfolio, we use equation \eqref{eq:CoCVaR Portfolio} and $\CoCVaR{\eta}{\Xi_p(w)}{\zeta}{\Xi_0}$ is obtained by a multiple-integral form \eqref{eq:CoCVaR Integral Form} or an expectation form \eqref{eq:CoCVaR Expectation Form}. 

Assume that we have an equally weighted portfolio $R_p(w)$ with the 30 stocks in Table \ref{table:ParamEst}, i.e. $w_n = 1/30$ for $n \in \{ 1,2,\cdots, 30\}$. 
We apply Proposition \ref{prop:Index and Portfolio} to find $\mu_p(w)$, $\sigma_p(w)$, $\beta_p(w)$, $\rho_p(w)$, $\Xi_p(w)$ and $\Xi_0(w)$.
Let $M$ be the sample size of the simulation and generate two sets of independent pairs of standard normal random numbers, 
$
S_0=\left\{(\epsilon_{0,m},\epsilon_{1,m})\sim \Phi\left(0,\Sigma_\epsilon \right) | m = 1,2,\cdots, M \right\}
$
with $\Sigma_\epsilon=\begin{footnotesize}\left(\begin{matrix}
1, & 0 \\ 0, &1
\end{matrix}\right)\end{footnotesize}$
 and generate one set of tempered stable subordinators, $T=\{\Tau_m|m = 1,2,\cdots, M\}$, which are independent of $S_0$.
Considering the correlation $\rho_p(w)$, we set $S_p=\{(\epsilon_{0,m},\epsilon_{p,m})|\epsilon_{p,m}=\rho_p(w)\epsilon_{0,m}+\sqrt{1-(\rho_p(w))^2}\epsilon_{1,m}, (\epsilon_{0,m}, \epsilon_{1,m})\in S_0\}$.

\begin{figure}[t]
\centering
\includegraphics[width = 12cm]{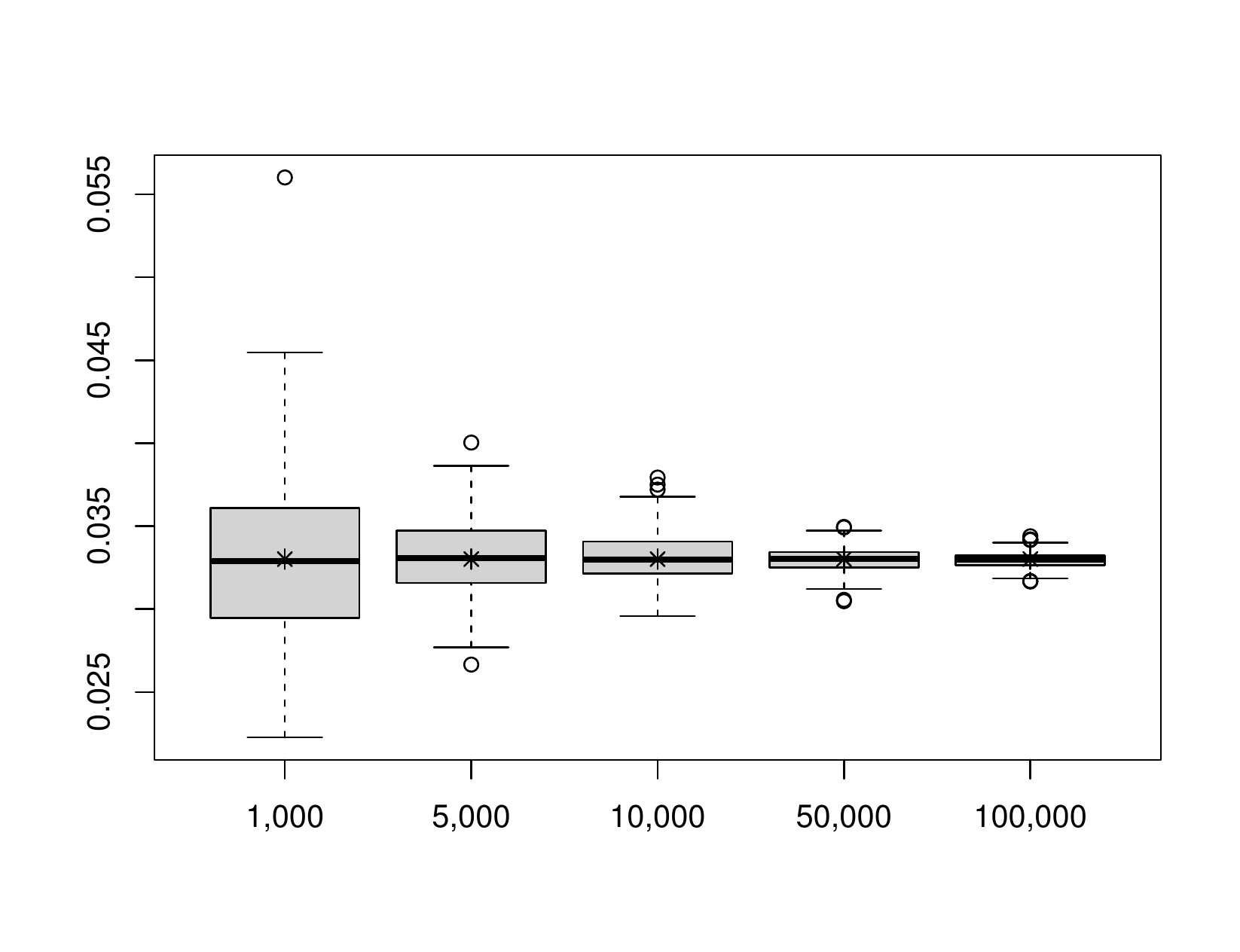}
\caption{\label{fig:BootStrapping_CoCVaR}Boot strapping for $\CoCVaR{\eta}{R_p(w)}{\zeta}{R_0}$. The x-axis is the sample size. Box plots for 100 MCS prices are presented on the plate for each sample size.}
\end{figure}

We obtain $\CoCVaR{\eta}{R_p(w)}{\zeta}{R_0}$ as
\begin{equation}\label{eq:mcs cocvar}
-\mu_p(w)-\frac{\sigma_p(w)}{\zeta\eta M}\sum_{m=1}^M\left(\beta_p(w)(\Tau_m-1)+\epsilon_{p,m}\gamma_p(w)\sqrt{\Tau_m}\right)1_{\epsilon_{p,m}<C(w,\Tau_m)}1_{\epsilon_0<v(\Tau_m)}
\end{equation}
for $(\epsilon_{0,m},\epsilon_{p,m})\in S_p$ and $\Tau_m\in T$, by MCS and \eqref{eq:CoCVaR Portfolio}.
We check the convergence of MCS using bootstrapping.
We draw the first boxplot from the left by repeating the MCS 100 times using \eqref{eq:mcs cocvar} for a given sample size $M=1000$ in Figure \ref{fig:BootStrapping_CoCVaR}. In addition, we draw the other boxplots for each sample size $M\in\{ 5000, 10000, 50000, 100000\}$, which are presented sequentially after the first box plot in the figure. 
We also show (by the symbol `*’) the CoCVaR computed by numerical integration based on \eqref{eq:CoCVaR Integral Form}. 
We observe that, as the sample size increases, the interquartile distance of MCS CoCVaR narrows, and dispersions are reduced. All box plots contain the CoCVaR computed by the numerical integration in the interquartile range.
The multiple-integral form \eqref{eq:CoCVaR Integral Form} takes a relatively longer numerical calculation time of 18.13 seconds,
while the MCS takes 7.57 seconds for a 100,000 sample size.

Using the same arguments, we can calculate MCT-CoVaR and MCT-CoCVaR using MCS.
We calculate the expectations in equations in Proposition \ref{prop:mctcovar} and \ref{prop:mctCoCVaR} use the random numbers in $S_p$, and $T$,
and substitute those MCS expectations into \eqref{eq:dGdx}, \eqref{eq:dGdwj}, and \eqref{eq:mctcocvar} to obtain $\frac{\partial}{\partial w_j}\CoVaR{\eta}{\Xi_p(w)}{\zeta}{\Xi_0}$ and $\frac{\partial}{\partial w_j}\CoCVaR{\eta}{\Xi_p(w)}{\zeta}{\Xi_0}$. Finally, we obtain $\frac{\partial}{\partial w_j}\CoVaR{\eta}{R_p(w)}{\zeta}{R_0}$ and $\frac{\partial}{\partial w_j}\CoCVaR{\eta}{R_p(w)}{\zeta}{R_0}$ by substituting $\frac{\partial}{\partial w_j}\CoVaR{\eta}{\Xi_p(w)}{\zeta}{\Xi_0}$ and $\frac{\partial}{\partial w_j}\CoCVaR{\eta}{\Xi_p(w)}{\zeta}{\Xi_0}$ into \eqref{eq:mct CoVaR Portfolio}, and \eqref{eq:mct CoCVaR Portfolio}, respectively.

The result of MCT-CoVaR and MCT-CoCVaR values using MCS for each individual stocks of the equally weighted portfolio $R_p(w)$ are exhibited in Table \ref{table:ParamEst} with  the rank of the values for the ascending order.
For example, MCT-CoVaR and MCT-CoCVaR values for MRK rank 1st and those values of AXP rank 30th, respectively.
DD ranks 2nd for MCT-CoVaR but 4th for MCT-CoCVaR and so on. We can see that AXP is largest CoVaR and CoCVaR contributor, and is recommended to reduce the proportion of the capital allocation. This idea will be discussed in the Risk Budgeting, below.

\subsection{Portfolio Optimization}
\begin{figure}[ht]
\centering
\includegraphics[width = 12cm]{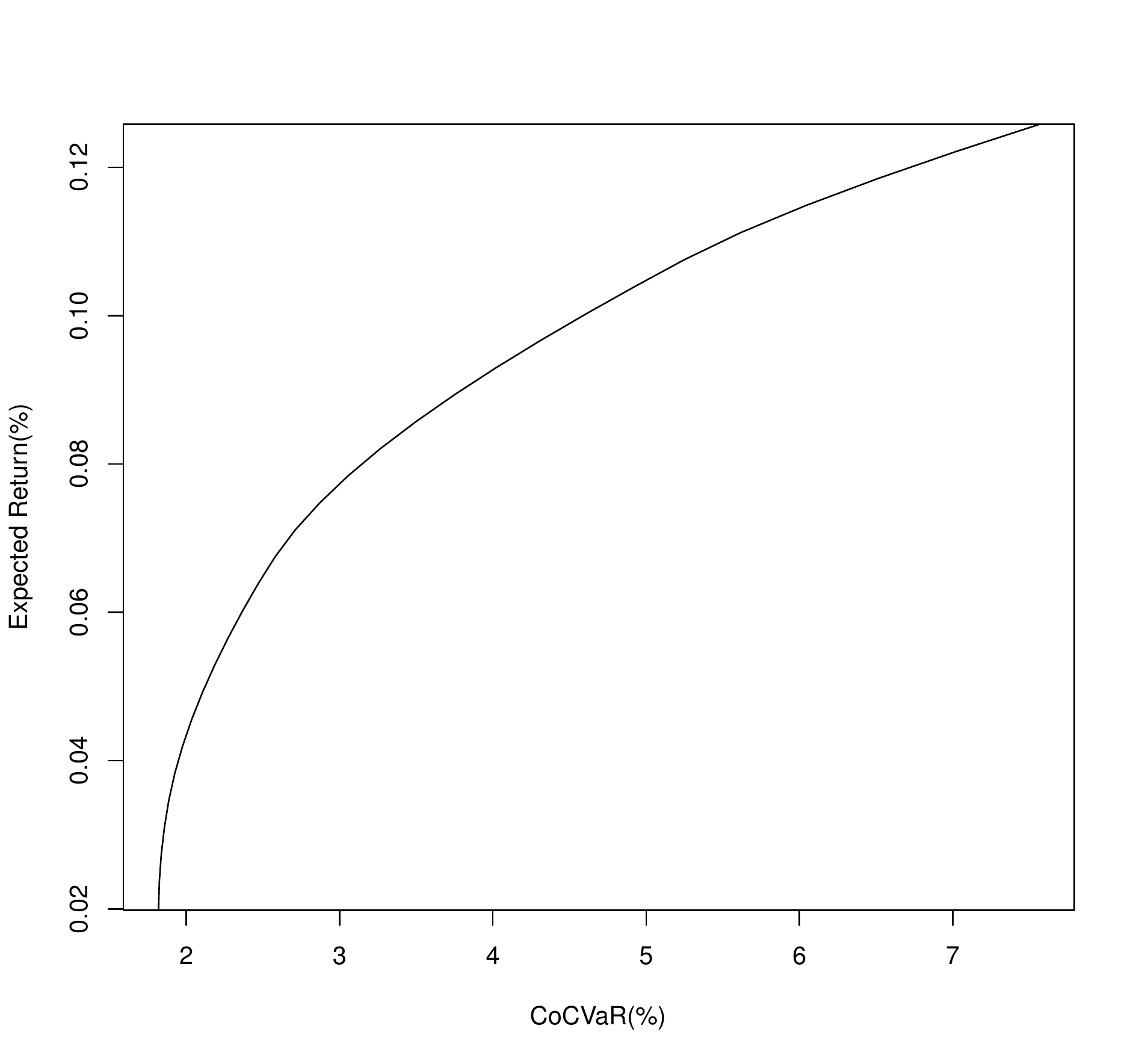}
\caption{\label{fig:Eff_CoCVaR}Efficient frontier}
\end{figure}
Since CoVaR and CoCVaR can capture the relative tail risk under the condition of distressed condition of a benchmark asset or index, we can use those two risk measures for relative portfolio optimization. In this section, we show an empirical example of CoCVaR minimizing portfolio optimization for the 30 stocks with respect to the DJIA index on the NTS market model.

We set a nonlinear programming problem for the portfolio optimization as
\begin{align*}
&\min_{w} \CoCVaR{\eta}{R_p(w)}{\zeta}{R_0}\\
\text{subject to}~~~
& w^\tr\mu \ge \mu^* \\
&\sum_{n=1}^N w_n = 1\\
& w_n\ge 0 \text{ for all } n\in \{1,2,\cdots, N\} 
\end{align*}
where the benchmark values for the portfolio expected return is
$\mu^*\in[\min(\mu),\max(\mu)]$. 
Using the parameters in Table \ref{table:ParamEst}, we perform the portfolio optimization for 51  points of $\mu^*$ in $\{\mu=\min(\mu)+k\cdot (\max(\mu)-\min(\mu))/50\,|\, k=0,1,2,\cdots, 50\}$. We finally obtain the efficient frontier in Figure \ref{fig:Eff_CoCVaR}.

\begin{table}[ht]
\centering
\begin{tabular}{c|ccc|c|rr|rr}
\hline
Symbol & $\mu_n (\%)$ & $\sigma_n (\%)$ & $\beta_n (\%)$  & KS  &     &      &                &       \\ 
\hline
DJIA & 0.0310 & 1.4575 & $-$3.7939 & 0.0100  & \footnotesize{MCT-CoVaR(\%)} & Rank & \footnotesize{MCT-CoCVaR(\%)} & Rank      \\
  \hline
AAPL & 0.1258 & 2.1977 & $-$1.2679 & 0.0587 & 1.6529 &    17 & 3.5957 &    19 \\ 
  AMGN & 0.0490 & 1.6810 & 1.8704 & 0.0416 & 3.2980 &    22 & 2.1241 &     7 \\ 
  AXP & 0.0395 & 2.5245 & 0.5258 & 0.0241 & 7.4288 &    30 & 6.6507 &    30 \\ 
  BA & $-$0.0564 & 3.4757 & 3.7589 & 0.0168 & 1.7783 &    18 & 6.3659 &    29 \\ 
  CAT & 0.0733 & 2.1609 & 0.3824 & 0.0502 & 5.5670 &    26 & 4.8958 &    22 \\ 
  CRM & 0.0284 & 2.5370 & $-$1.0588 & 0.0514 & 5.1597 &    23 & 5.3903 &    27 \\ 
  CSCO & 0.0117 & 1.9311 & $-$1.3543 & 0.0323 & 1.5658 &    15 & 2.2187 &     8 \\ 
  CVX & 0.0678 & 2.4216 & 1.5941 & 0.0176 & 5.6102 &    27 & 4.9429 &    23 \\ 
  DD & $-$0.0096 & 2.4072 & 0.3164 & 0.0593 & $-$0.1376 &     2 & 1.4762 &     4 \\ 
  DIS & $-$0.0140 & 2.1879 & 4.1225 & 0.0337 & 0.3701 &    10 & 3.4376 &    18 \\ 
  GS & 0.0791 & 2.2039 & 2.0020 & 0.0485 & 6.0133 &    29 & 5.4714 &    28 \\ 
  HD & 0.0709 & 1.9351 & $-$3.9629 & 0.0261 & 5.3208 &    25 & 4.9610 &    24 \\ 
  HON & 0.0483 & 1.8365 & $-$2.3242 & 0.0341 & 1.5242 &    14 & 2.6857 &    13 \\ 
  IBM & 0.0437 & 1.8054 & $-$1.5555 & 0.0287 & 0.5364 &    11 & 3.1157 &    16 \\ 
  INTC & $-$0.0314 & 2.4868 & $-$2.8086 & 0.0295 & 0.2124 &     8 & 2.6130 &    11 \\ 
  JNJ & 0.0260 & 1.3628 & $-$0.5663 & 0.0285 & 0.8420 &    12 & 1.5451 &     5 \\ 
  JPM & 0.0333 & 2.1622 & 2.5850 & 0.0357 & 5.6910 &    28 & 5.1918 &    26 \\ 
  KO & 0.0324 & 1.4671 & $-$3.6931 & 0.0249 & 0.3419 &     9 & 2.6197 &    12 \\ 
  MCD & 0.0476 & 1.5731 & 1.2955 & 0.0162 & 0.2091 &     7 & 2.3540 &     9 \\ 
  MMM & $-$0.0286 & 1.7824 & $-$2.1382 & 0.0420 & 0.0681 &     6 & 2.7091 &    14 \\ 
  MRK & 0.0461 & 1.5310 & $-$0.0218 & 0.0463 & $-$0.7418 &     1 & 0.2360 &     1 \\ 
  MSFT & 0.0879 & 2.0241 & $-$2.1377 & 0.0469 & 2.4863 &    20 & 3.9390 &    20 \\ 
  NKE & 0.0424 & 2.1704 & $-$0.2999 & 0.0418 & 2.4980 &    21 & 4.5851 &    21 \\ 
  PG & 0.0511 & 1.4348 & $-$3.0200 & 0.0322 & $-$0.1305 &     3 & 0.7501 &     3 \\ 
  TRV & 0.0444 & 1.9422 & $-$2.8043 & 0.0255 & 1.1151 &    13 & 3.3984 &    17 \\ 
  UNH & 0.0714 & 1.9919 & 1.6187 & 0.0288 & 0.0165 &     5 & 2.5646 &    10 \\ 
  V & 0.0472 & 1.9411 & $-$2.5133 & 0.0423 & 5.2771 &    24 & 5.1686 &    25 \\ 
  VZ & $-$0.0226 & 1.2645 & $-$2.5648 & 0.0357 & $-$0.1084 &     4 & 0.3166 &     2 \\ 
  WBA & $-$0.0539 & 2.2096 & $-$1.4927 & 0.0360 & 1.6265 &    16 & 2.7307 &    15 \\ 
  WMT & 0.0508 & 1.4940 & 1.0659 & 0.0263 & 2.4613 &    19 & 1.7390 &     6 \\ 
   \hline
\end{tabular}
\caption{\label{table:ParamEst}NTS parameter fit using 1-day-returns from 11/27/2018 to 11/15/2022. $\alpha = 1.1835$, $\theta = 0.0820$}
\end{table}

\subsection{Risk Budgeting}

The MCT-CoVaR and MCT-CoCVaR allow portfolio managers to decide to rebalance their capital allocation weights. Managers can reduce portfolio risk by decreasing the weight of high-risk contributors and increasing the weight of low-risk contributors. The high-risk contributors are assets having high MCT-CoVaR (or MCT-CoCVaR), while the low-risk contributors are assets having low MCT-CoVaR (or MCT-CoCVaR). 

Consider a capital allocation vector $w$ for a portfolio with $N$ member stocks. Let $\varDelta w=(\varDelta w_1$, $\varDelta w_2$, $\cdots$, $\varDelta w_N)^\tr\in D$ where $D$ is a zero neighborhood in $\R^N$, and let 
\[
\varDelta \CoVaR{\eta}{R_P(w)}{\zeta}{R_0}=\CoVaR{\eta}{R_P(w+\varDelta w)}{\zeta}{R_0}-\CoVaR{\eta}{R_P(w)}{\zeta}{R_0},
\]
and
\[
\varDelta \CoCVaR{\eta}{R_P(w)}{\zeta}{R_0}=\CoCVaR{\eta}{R_P(w+\varDelta w)}{\zeta}{R_0}-\CoCVaR{\eta}{R_P(w)}{\zeta}{R_0}.
\]
The optimal portfolios with respect to CoVaR and CoCVaR are obtained by solving the following problem:
\begin{align}
&\min_{\varDelta w} \varDelta \CoVaR{\eta}{R_P(w)}{\zeta}{R_0}  \\
\nonumber
&\text{subject to }  E[R_p(w+\varDelta w)]-E[R_P(w)]\ge 0
\text{ and } \sum_{j=1}^N \varDelta w_j = 0,
\end{align}
and 
\begin{align}
&\min_{\varDelta w} \varDelta \CoCVaR{\eta}{R_P(w)}{\zeta}{R_0}  \\
\nonumber
&\text{subject to }  E[R_p(w+\varDelta w)]-E[R_P(w)]\ge 0
\text{ and } \sum_{j=1}^N \varDelta w_j = 0.
\end{align}
Since we have
\begin{align*}
&\varDelta \CoVaR{\eta}{R_P(w)}{\zeta}{R_0}\approx\sum_{j=1}^N \left(\frac{\partial}{\partial w_j}\CoVaR{\eta}{R_P(w)}{\zeta}{R_0} \right) \varDelta w_j,
\\
&\varDelta \CoCVaR{\eta}{R_P(w)}{\zeta}{R_0}\approx\sum_{j=1}^N \left(\frac{\partial }{\partial w_j}\CoCVaR{\eta}{R_P(w)}{\zeta}{R_0} \right) \varDelta w_j \\
&\text{ and }  E[R_p(w+\varDelta w)]-E[R_P(w)]=\sum_{j=1}^N \mu_j\varDelta w_j,
\end{align*}
we can find the optimal portfolio on the local domain $D$ with respect to CoVaR and CoCVaR, respectively, as follows:
\begin{align}\label{eq:OptPortVaR}
&\varDelta w^*=\argmin_{\varDelta w} \sum_{j=1}^N \left(\frac{\partial}{\partial w_j}\CoVaR{\eta}{R_P(w)}{\zeta}{R_0} \right) \varDelta w_j \\
\nonumber
&\text{subject to }  \sum_{j=1}^N \mu_j\varDelta w_j\ge 0
\text{ and } \sum_{j=1}^N \varDelta w_j = 0.
\end{align}
and 
\begin{align}
\label{eq:OptPortAVaR}
&\varDelta w^*=\argmin_{\varDelta w} \sum_{j=1}^N \left(\frac{\partial}{\partial w_j}\CoCVaR{\eta}{R_P(w)}{\zeta}{R_0} \right) \varDelta w_j \\
\nonumber
&\text{subject to }  \sum_{j=1}^N \mu_j\varDelta w_j\ge 0
\text{ and } \sum_{j=1}^N \varDelta w_j = 0.
\end{align}

We perform the risk budgeting for CoVaR and CoCVaR using the 30 stocks in Table \ref{table:DJIA Members}, i.e. $N=30$, with the estimated parameters in Table \ref{table:ParamEst}, iteratively, as the following algorithm:
\begin{enumerate}
\item[] \textbf{Step 1.} Generate a set of independent pairs of standard normal random numbers $S_0$, and a set of tempered stable subordinators, $T$, which are independent of $S_0$, as we discussed in Section \ref{sec:mcs CoVaR CoCVaR}. Here, the sample size is $M=100,000$.
\item[] \textbf{Step 2.} Select an initial capital allocation weight vector $w$.
\item[] \textbf{Step 3.} 
Find $\mu_p(w)$, $\sigma_p(w)$, $\beta_p(w)$, and $\rho_p(w)$ by Proposition \ref{prop:Index and Portfolio}.
\item[] \textbf{Step 4.} 
Regenerate a set of bivariate standard normal random vectors $S_p=\{(\epsilon_{0,m},\epsilon_{p,m})|\epsilon_{p,m}=\rho_p(w)\epsilon_{0,m}+\sqrt{1-(\rho_p(w))^2}\epsilon_{1,m}, (\epsilon_{0,m}, \epsilon_{1,m})\in S_0\}$ with correlation $\rho_p(w)$.
\item[] \textbf{Step 5.} Calculate MCT-CoVaR or MCT-CoCVaR for $w$ using the MCS method we discussed in Section \ref{sec:mcs CoVaR CoCVaR}.
\item[] \textbf{Step 6.} Perform risk budgeting and find $\varDelta w^*$ using \eqref{eq:OptPortVaR}, or using \eqref{eq:OptPortAVaR}, for for the CoVaR risk budgeting, or the CoCVaR risk budgeting, respectively. 
\item[] \textbf{Step 7.} Change $w$ to $w+\varDelta w^*$ and go to Step 3. Repeat [Step 2 - Step 7] $L$ times.
\end{enumerate}  
Let the initial capital allocation weight vector be equally weighted. 
We perform the iterative risk budgeting $L=200$ times for the local domain be
\[
D=\{(x_1,x_2,\cdots,x_{30})\,|\, x_j\in[-4\cdot 10^{-4},4\cdot 10^{-4}],\, j=1,2,\cdots, 30 \}.
\]
The results are exhibited in Figure \ref{fig:RiskBudgeting_CoCVaR}. For each iteration, we show the values of CoVaR, and CoCVaR with the MCS method in the left and right plate, respectively.
The figure shows that the portfolio CoVaR and CoCVaR with respect to the DJIA index decreases as increasing the number of iterations. 
That is using risk budgeting of CoVaR and CoCVaR on the NTS market model, 
we obtain a portfolio having the same expected return but less relative tail risk.

\begin{figure}[ht]
\centering
\includegraphics[width = 8cm]{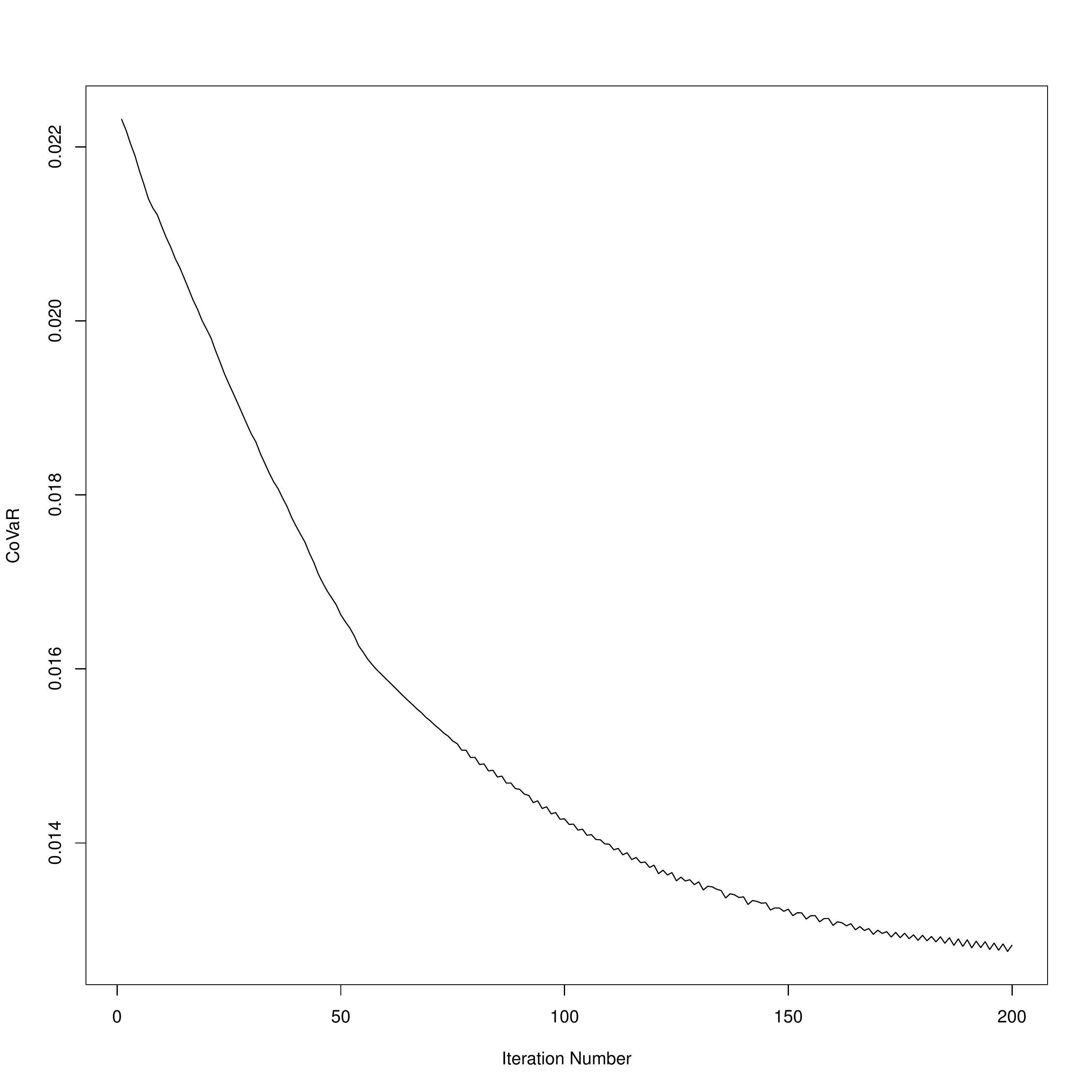}
\includegraphics[width = 8cm]{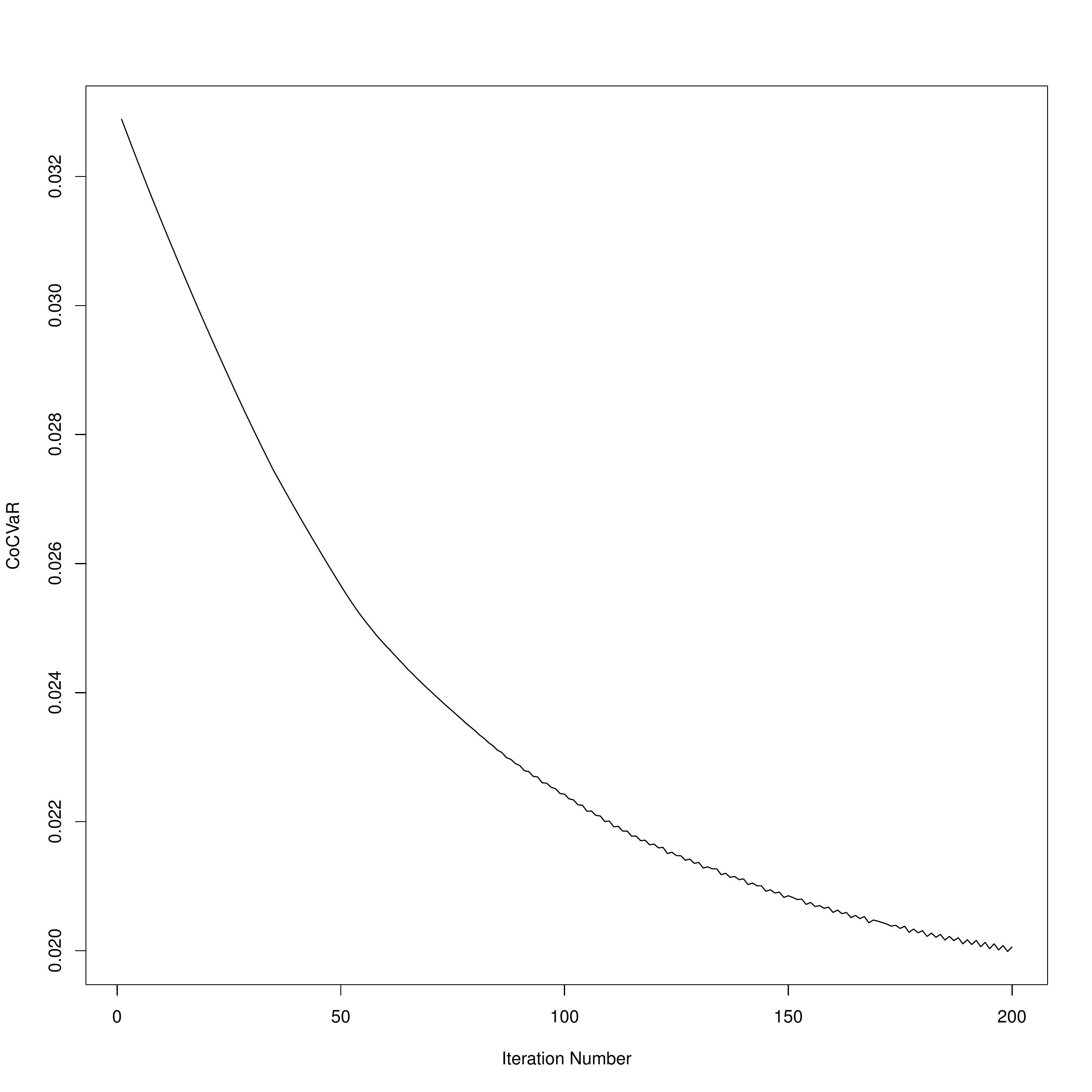}
\caption{\label{fig:RiskBudgeting_CoCVaR}Risk Budgeting iteration}
\end{figure}

\section{conclusion}
This paper presents portfolio CoVaR and portfolio CoCVaR on the NTS market model. We develop an MCS method to calculate portfolio CoVaR and CoCVaR, and apply it to portfolio optimization. As an empirical illustration, we consider a portfolio consisting of 30 stocks and measure the CoCVaR of the portfolio with respect to the DJIA index, and then find the efficient frontiers of the portfolio, maximizing the portfolio's expected return and minimizing the relative tail risk captured by the CoCVaR.
In addition, we find an analytic formula for the marginal contributions to CoVaR and CoCVaR, which we calculate using the MCS method to overcome numerical difficulties. 
We also perform portfolio risk budgeting methods using the marginal contributions to CoVaR and CoCVaR on the NTS market model. We empirically show that the portfolio CoVaR and CoCVaR are decreased by using portfolio budgeting iteratively.

\section{Appendix}
\begin{proof}[Proof of Proposition \ref{prop:Index and Portfolio}]
By Proposition \ref{pro:mu+sigma stdNTS}, the stock portfolio return $R_p(w) = \sum_{n=1}^N w_n R_n$ is equal to
\[
R_p(w) = \mu_p(w) + \sigma_p(w) \Xi_p(w), \text{ with } \Xi_p(w)\sim \textup{stdNTS}_{1}(\alpha, \theta, \beta_p(w), 1),
\]
where
\[
\mu_p(w) = \sum_{n=1}^N w_n \mu_n,~~~ 
\beta_p(w) = \frac{1}{\sigma_p(w)}\sum_{n=1}^N  w_n\sigma_n\beta_n.
\]
and
\begin{equation}\label{eq:sigma_p(w)} 
\sigma_p(w) 
= \sqrt{\sum_{n=1}^N\sum_{m=1}^N w_n w_m \sigma_n \sigma_m \cov\left(\Xi_n, \Xi_m\right)}.
\end{equation}
Since 
$
\Xi_p(w) = \frac{1}{\sigma_p(w)}\sum_{n=1}^N w_n \sigma_n \Xi_n,
$
we have
\begin{align*}
\Xi_p(w) &= \frac{1}{\sigma_p(w)}\sum_{n=1}^N \left(w_n \sigma_n \beta_n\left(\Tau-1\right)+\gamma_n\sqrt{\Tau}\varepsilon_n \right)
\\&=\beta_p(w)\left(\Tau-1\right)+\frac{\sum_{n=1}^N w_n \sigma_n \gamma_n \varepsilon_n}{\sigma_p(w)} \sqrt{\Tau}
\end{align*}
with $\gamma_n = \sqrt{1-\beta_n^2\left(\frac{2-\alpha}{2\theta}\right)}$.
Note that we have
\[
E\left[\frac{\sum_{n=1}^N w_n \sigma_n \gamma_n \varepsilon_n}{\sigma_p(w)}\right]=0
\]
and
\begin{align*}
\var\left(\frac{\sum_{n=1}^N w_n \sigma_n \gamma_n \varepsilon_n}{\sigma_p(w)}\right) &= \frac{1}{\sigma_p(w)}\sum_{n=1}^N\sum_{m=1}^Nw_nw_m\sigma_n\sigma_m\gamma_n\gamma_n\rho_{n,m}.
\end{align*}
By \eqref{eq:cov Xi_n and Xi_m} and  \eqref{eq:sigma_p(w)}, we get
\begin{align*}
\var\left(\frac{\sum_{n=1}^N w_n \sigma_n \gamma_n \varepsilon_n}{\sigma_p(w)}\right) &=\frac{1}{(\sigma_p(w))^2}\sum_{n=1}^N\sum_{m=1}^N w_n w_m\sigma_n\sigma_m\left(\cov(\Xi_n,\Xi_m)-\beta_n\beta_m\left(\frac{2-\alpha}{2\theta}\right)\right)
\\
&=1-\frac{1}{(\sigma_p(w))^2}\left(\sum_{n=1}^n w_n\sigma_n\beta_n\right)^2\left(\frac{2-\alpha}{2\theta}\right)
\\&=1-\left(\beta_p(w)\right)^2\left(\frac{2-\alpha}{2\theta}\right)
\end{align*}
For the Gaussian property, we get
\[
\frac{\sum_{n=1}^N w_n \sigma_n \gamma_n \varepsilon_n}{\sigma_p(w)} = \sqrt{\var
\left(
\frac{\sum_{n=1}^N w_n \sigma_n \gamma_n \varepsilon_n}{\sigma_p(w)}
\right)
}
\varepsilon_p = \sqrt{1-\left(\beta_p(w)\right)^2\left(\frac{2-\alpha}{2\theta}\right)}\varepsilon_p,
\]
where $\varepsilon_p\sim \Phi(0,1)$. Hence, we have
\[
\Xi_p(w) = \beta_p(w)\left(\Tau-1\right)+\gamma_p(w)\sqrt{\Tau}\varepsilon_p
\]
where $\gamma_p(w)=\sqrt{1-\left(\beta_p(w)\right)^2\left(\frac{2-\alpha}{2\theta}\right)}$.
Therefore, we have
\begin{equation}\label{eq:cov(xi0,xip)1}
\cov\left(\Xi_0, \Xi_p(w)\right) = \gamma_0\gamma_p(w)\rho_{0,p}+\beta_0\beta_p(w)\left(\frac{2-\alpha}{2\theta}\right).
\end{equation}
On the other hand, we have
\begin{align*}
\cov\left(\Xi_0, \Xi_p(w)\right) &= \cov\left(\beta_0(\Tau-1)+\gamma_0\sqrt{\Tau}\varepsilon_0, \frac{1}{\sigma_p(w)}\sum_{n=1}^Nw_n \sigma_n \left(\beta_n(\Tau-1)+\gamma_n\sqrt{\Tau}\varepsilon_n\right)\right)\\
&=\frac{1}{\sigma_p(w)}\sum_{n=1}^Nw_n \sigma_n \cov\left(\beta_0(\Tau-1)+\gamma_0\sqrt{\Tau}\varepsilon_0, \beta_n(\Tau-1)+\gamma_n\sqrt{\Tau}\varepsilon_n \right)
\\
&=\frac{1}{\sigma_p(w)}\sum_{n=1}^Nw_n \sigma_n E\left[ \left(\beta_0(\Tau-1)+\gamma_0\sqrt{\Tau}\varepsilon_0\right)\left(\beta_n(\Tau-1)+\gamma_n\sqrt{\Tau}\varepsilon_n\right) \right].
\end{align*}
Since we have
\begin{align*}
&E\left[\left(\beta_n(\Tau-1)+\gamma_n\sqrt{\Tau}\varepsilon_n\right) \left(\beta_0(\Tau-1)+\gamma_0\sqrt{\Tau}\varepsilon_0\right) \right]
\\
&=
E\left[\beta_n\beta_0(\Tau-1)^2\right]+E\left[\gamma_n\gamma_0\Tau\varepsilon_n\varepsilon_0\right]\\
&=\beta_n\beta_0 \var(\Tau)+\gamma_n\gamma_0\rho_{0,n}=\beta_n\beta_0\left(\frac{2-\alpha}{2\theta}\right)+\gamma_n\gamma_0\rho_{0,n},
\end{align*}
we obtain
\begin{align}
\nonumber
\cov\left(\Xi_0, \Xi_p(w)\right) &= \frac{1}{\sigma_p(w)}\sum_{n=1}^Nw_n \sigma_n\left(\beta_n\beta_0\left(\frac{2-\alpha}{2\theta}\right)+\gamma_n\gamma_0\rho_{0,n}\right)
\\
\label{eq:civ(xi0,xip)2}
&=
\beta_0 \beta_p(w) \left(\frac{2-\alpha}{2\theta}\right)+\frac{\gamma_0}{\sigma_p(w)}\sum_{n=1}^Nw_n \sigma_n\gamma_n\rho_{0,n}.
\end{align}
By \eqref{eq:cov(xi0,xip)1} and \eqref{eq:civ(xi0,xip)2}, we have
\begin{align*}
\gamma_0\gamma_p(w)\rho_{0,p}+\beta_0\beta_p(w)\left(\frac{2-\alpha}{2\theta}\right)
=\beta_0 \beta_p(w)\left(\frac{2-\alpha}{2\theta}\right)+\frac{\gamma_0}{\sigma_p(w)}\sum_{n=1}^Nw_n \sigma_n\gamma_n\rho_{0,n}
\end{align*}
or
\begin{align}\label{eq:rho0p dummy}
\rho_{0,p} = \frac{1}{\sigma_p(w)\gamma_p(w)}\sum_{n=1}^Nw_n \sigma_n\gamma_n\rho_{0,n}.
\end{align}
Using the definition $\gamma_p(w)=\sqrt{1-\left(\beta_p(w)\right)^2 \left(\frac{2-\alpha}{2\theta}\right) }$, we obtain
\begin{align}
\nonumber
\sigma_p(w)\gamma_p(w)
&= \sqrt{\left(\sigma_p(w)\right)^2 -\left(\sigma_p(w)\right)^2\left(\beta_p(w)\right)^2\left(\frac{2-\alpha}{2\theta}\right)}
\\
\nonumber
&=
 \sqrt{\sum_{n=1}^N\sum_{m=1}^Nw_n w_m \sigma_n \sigma_m \cov\left(\Xi_n,\Xi_m\right) -\left(\sum_{n=1}^Nw_n\sigma_n\beta_n\right)^2\left(\frac{2-\alpha}{2\theta}\right)}
\\
\label{eq:sigma times gamma}
&=\sqrt{\sum_{n=1}^N\sum_{m=1}^Nw_n w_m \sigma_n \sigma_m  \gamma_n \gamma_m \rho_{n,m} }.
\end{align}
By substituting \eqref{eq:sigma times gamma} into \eqref{eq:rho0p dummy}, we obtain \eqref{eq:rho 0,p}.
\end{proof}

\begin{lemma}
Based on the setting in Section \ref{Sec:MCTCoCVaR}, we have\\
(i)
\begin{align}\label{eq:d sigma d wj}
\frac{\partial}{\partial w_j}\sigma_p(w) 
=\frac{1}{\sigma_p(w)}\sum_{n=1}^N w_n\sigma_n\sigma_j\cov(\Xi_n,\Xi_j)
\end{align}
(ii)
\begin{align}\label{eq:d beta d wj} 
\frac{\partial}{\partial w_j}\beta_p(w) 
=
\frac{\sigma_j}{\sigma_p(w)}\left(
\beta_j 
-
\frac{\beta_p(w)}{\sigma_p(w)}\sum_{n=1}^N w_n\sigma_n\cov(\Xi_n,\Xi_j)
\right)
\end{align}
and
\begin{align}\label{eq:d gamma d wj}
\frac{\partial}{\partial w_j}\gamma_p(w) 
=-\frac{\beta_p(w)\left(\frac{2\theta}{2-\alpha}\right)}{\gamma_p(w)}\frac{\partial}{\partial w_j}\beta_p(w) 
\end{align}
(iii)
\begin{align}\label{eq:d u d wj}
&\frac{\partial}{\partial w_j}u(x,w,t)
\\
\nonumber
&
=\frac{\sigma_j}{\sigma_p(w)}%
\left(\frac{1-t}{\gamma_p(w)\sqrt{t}}
+\frac{u(x,w,t)\beta_p(w)\left(\frac{2\theta}{2-\alpha}\right)}{(\gamma_p(w))^2}\right)
\left(
\beta_j-\frac{\beta_p(w)}{\sigma_p(w)}\sum_{n=1}^N w_n\sigma_n\cov(\Xi_n,\Xi_j)
\right)
\end{align}
(iv) Let
\begin{align*}
C(w,t) = \frac{ -\CoVaR{\eta}{\Xi_p(w)}{\zeta}{\Xi_0}-\beta_p(w)(t-1)}{\gamma_p(w)\sqrt{t}}
\end{align*}
Then
\begin{align}\label{eq:dCdwj}
 \frac{\partial}{\partial w_j}C(w,t)
&=-\frac{1}{t(\gamma_p(w))^2}\frac{\partial}{\partial w_j}\CoVaR{\eta}{\Xi_p(w)}{\zeta}{\Xi_0}
\\
&
\nonumber
+\frac{\sigma_j}{\sigma_p(w)}%
\left(\frac{1-t}{\gamma_p(w)\sqrt{t}}
+\frac{C(w,t)\beta_p(w)\left(\frac{2\theta}{2-\alpha}\right)}{(\gamma_p(w))^2}\right)
\left(
\beta_j-\frac{\beta_p(w)}{\sigma_p(w)}\sum_{n=1}^N w_n\sigma_n\cov(\Xi_n,\Xi_j)
\right)
%
\end{align}
(v)
\begin{align}\label{eq:d rho d wj}
\frac{\partial}{\partial w_j}\rho_p(w)=
\gamma_j \sigma_j\left(\frac{ \rho_{0,j} }{\sqrt{w^\tr \Sigma_{\gamma,\sigma}^* w}}
-\frac{\rho_p(w)}{w^\tr \Sigma_{\gamma,\sigma}^* w}\sum_{n=1}^N w_n \gamma_n \sigma_n \rho_{n,j}\right)
\end{align}
\end{lemma}
\begin{proof}
(i)
\begin{align*}
\frac{\partial}{\partial w_j}\sigma_p(w) 
&
=\frac{1}{2\sigma_p(w)}\sum_{n=1}^N 2w_n\sigma_n\sigma_j\cov(\Xi_n,\Xi_j)
\\
&
=\frac{1}{\sigma_p(w)}\sum_{n=1}^N w_n\sigma_n\sigma_j\left(\gamma_n\gamma_j\rho_{n,j}+\beta_n\beta_j\left(\frac{2-\alpha}{2\theta}\right)\right)
\end{align*}

(ii)
\begin{align*}
\frac{\partial}{\partial w_j}\beta_p(w) 
&= \frac{1}{\left(\sigma_p(w)\right)^2}\left(
\sigma_j \beta_j \sigma_p(w)
-\left(\sum_{n=1}^N w_n \sigma_n \beta_n \right)
\frac{\partial}{\partial w_j}\sigma_p(w) 
\right) \\
&= \frac{1}{\left(\sigma_p(w)\right)^2}\left(
\sigma_j \beta_j \sigma_p(w)
-\sigma_p(w)\beta_p(w)
\frac{\partial}{\partial w_j}\sigma_p(w) 
\right) \\
&=
\frac{\sigma_j}{\sigma_p(w)}\left(
\beta_j 
-
\frac{\beta_p(w)}{\sigma_p(w)}\sum_{n=1}^N w_n\sigma_n\cov(\Xi_n,\Xi_j)
\right) \\
\end{align*}
and
\begin{align*}
\frac{\partial}{\partial w_j}\gamma_p(w) 
&=\frac{-2\beta_p(w)\left(\frac{2\theta}{2-\alpha}\right)}{2\sqrt{1-(\beta_p(w))^2\left(\frac{2\theta}{2-\alpha}\right)}}\frac{\partial}{\partial w_j}\beta_p(w) 
=-\frac{\beta_p(w)\left(\frac{2\theta}{2-\alpha}\right)}{\gamma_p(w)} \frac{\partial}{\partial w_j}\beta_p(w) 
\end{align*}

(iii) We have
\begin{align}
\nonumber
&\frac{\partial}{\partial w_j}u(x,w,t)\\
\nonumber
&=\frac{1}{t(\gamma_p(w))^2}
\Bigg(
\left(-(t-1)\frac{\partial}{\partial w_j}\beta_p(w)\right)\gamma_p(w)\sqrt{t}
-(x-(t-1)\beta_p(w))\sqrt{t}\frac{\partial}{\partial w_j}\gamma_p(w)
\Bigg)\\
&
=\left(\frac{1-t}{\gamma_p(w)\sqrt{t}}
+\frac{u(x,w,t)\beta_p(w)\left(\frac{2\theta}{2-\alpha}\right)}{(\gamma_p(w))^2}\right)\frac{\partial}{\partial w_j}\beta_p(w).
\label{eq:dudwj sub}
\end{align}
We obtain \eqref{eq:d u d wj}, by substituting \eqref{eq:d beta d wj} into the last equation.

(iv) We have
\begin{align*}
&\frac{\partial}{\partial w_j}C(w,t)\\
&=\frac{1}{t(\gamma_p(w))^2}
\Bigg(
\left(-\frac{\partial}{\partial w_j}\CoVaR{\eta}{\Xi_p(w)}{\zeta}{\Xi_0}-(t-1)\frac{\partial}{\partial w_j}\beta_p(w)\right)\gamma_p(w)\sqrt{t}
\\&
-\left(-\CoVaR{\eta}{\Xi_p(w)}{\zeta}{\Xi_0}-(t-1)\beta_p(w)\right)\sqrt{t}\frac{\partial}{\partial w_j}\gamma_p(w)
\Bigg)\\
&
=-\frac{1}{t(\gamma_p(w))^2}\frac{\partial}{\partial w_j}\CoVaR{\eta}{\Xi_p(w)}{\zeta}{\Xi_0} + \left(\frac{1-t}{\gamma_p(w)\sqrt{t}}
+\frac{C(w,t)\beta_p(w)\left(\frac{2\theta}{2-\alpha}\right)}{(\gamma_p(w))^2}\right)\frac{\partial}{\partial w_j}\beta_p(w).
\end{align*}
By substituting $C(w,t)=u(x,w,t)|_{x=-\CoVaR{\eta}{\Xi_p(w)}{\zeta}{\Xi_0}}$ into \eqref{eq:dudwj sub}, we have
\begin{align*}
\frac{\partial}{\partial w_j}C(w,t)
&
=-\frac{1}{t(\gamma_p(w))^2}\frac{\partial}{\partial w_j}\CoVaR{\eta}{\Xi_p(w)}{\zeta}{\Xi_0} 
+\frac{\partial}{\partial w_j}u(x,w,t)\Big|_{x=-\CoVaR{\eta}{\Xi_p(w)}{\zeta}{\Xi_0}}.
\end{align*}

(v)
\begin{align*}
&\frac{\partial}{\partial w_j}\rho_p(w)
\\
&=\frac{1}{w^\tr \Sigma_{\gamma,\sigma}^* w}\left(\gamma_j \sigma_j \rho_{0,j} \sqrt{w^\tr \Sigma_{\gamma,\sigma}^* w}-\frac{w^\tr V_{\gamma,\sigma, \rho}^*}{2\sqrt{w^\tr \Sigma_{\gamma,\sigma}^* w}}\gamma_j\sigma_j\sum_{n=1}^N 2w_n \gamma_n \sigma_n \rho_{n,j}\right)\\
&=
\gamma_j \sigma_j\left(\frac{ \rho_{0,j} }{\sqrt{w^\tr \Sigma_{\gamma,\sigma}^* w}}
-\frac{\rho_p(w)}{w^\tr \Sigma_{\gamma,\sigma}^* w}\sum_{n=1}^N w_n \gamma_n \sigma_n \rho_{n,j} \right)
\end{align*}
\end{proof}

\begin{lemma}
Let $f^{\Phi_2}_\rho$ be the pdf of the bivariate standard normal distribution with covariance $\rho$, then we have
\begin{align}\label{eq:2dimNormal to 1dimNormal}
\int_0^q f^{\Phi_2}_\rho(x, K)dx = \frac{1}{\sqrt{2\pi}}\exp\left(-\frac{K^2}{2}\right) F_{\Phi}\left(\frac{q-\rho K}{\sqrt{1-\rho^2}}\right).
\end{align}
Moreover, if we put $\rho = \rho_p(w)$ and let $w_j$ be the $j$-th element of $w\in I^N$ then we have
\begin{align}\label{eq:d f 2-dim normal dy}
\nonumber
&\frac{\partial }{\partial w_j}f^{\Phi_2}_{\rho_p(w)}(x_1, x_2) 
\\
&=\left(\frac{\rho_p(w)}{1-(\rho_p(w))^2}-\frac{\rho_p(w)x_1^2-(1+(\rho_p(w))^2)x_1x_2+\rho_p(w)x_2^2}{(1-(\rho_p(w))^2)^2}\right)f^{\Phi_2}_{\rho_p(w)} (x_1, x_2)\frac{\partial}{\partial w_j}\rho_p(w)
\end{align}
\end{lemma}

\begin{proof}
We have
\[
f^{\Phi_2}_\rho(x_1, x_2)=\frac{1}{2\pi\sqrt{1-\rho^2}}\exp\left(-\frac{x_1^2-2\rho x_1 x_2 + x_2^2}{2(1-\rho^2)}\right),
\]
and hence we obtain
\begin{align*}
&\int_{-\infty}^q f^{\Phi_2}_\rho(x, K) dx\\
&=\int_{-\infty}^q \frac{1}{2\pi\sqrt{1-\rho^2}}\exp\left(-\frac{x^2-2\rho x K + K^2}{2(1-\rho^2)}\right)dx\\
&=\int_{-\infty}^q \frac{1}{2\pi\sqrt{1-\rho^2}}\exp\left(-\frac{(x-\rho K)^2}{2(1-\rho^2)}\right)dx \exp\left(-\frac{K^2}{2}\right)\\
&= \frac{1}{\sqrt{2\pi}}\exp\left(-\frac{K^2}{2}\right) F_{\Phi}\left(\frac{q-\rho K}{\sqrt{1-\rho^2}}\right)
\end{align*}
which is \eqref{eq:2dimNormal to 1dimNormal}.
Moreover, we have
\begin{align*}
&\frac{\partial }{\partial w_j}f^{\Phi_2}_{\rho_p(w)} (x_1, x_2)\\
&
=
\frac{\partial}{\partial w_j}\left(\frac{1}{2\pi\sqrt{1-(\rho_p(w))^2}}\exp\left(-\frac{x_1^2-2\rho_p(w) x_1 x_2+x_2^2}{2(1-(\rho_p(w))^2)}\right)\right)\\
 &=\frac{1}{2\pi}\exp\left(-\frac{x_1^2-2\rho_p(w) x_1 x_2+x_2^2}{2(1-(\rho_p(w))^2)}\right)\frac{\partial}{\partial w_j}\frac{1}{\sqrt{1-(\rho_p(w))^2}} \\
 &+\frac{1}{2\pi\sqrt{1-(\rho_p(w))^2}}\frac{\partial}{\partial w_j}\exp\left(-\frac{x_1^2-2\rho_p(w) x_1 x_2+x_2^2}{2(1-(\rho_p(w))^2)}\right) \\
&=\frac{1}{2\pi}\exp\left(-\frac{x_1^2-2\rho_p(w) x_1 x_2+x_2^2}{2(1-(\rho_p(w))^2)}\right)\left(-\frac{1}{2}\right)\left(1-(\rho_p(w))^2\right)^{-\frac{3}{2}}(-2\rho_p(w)) \frac{\partial}{\partial w_j}\rho_p(w)\\
&+\frac{1}{2\pi\sqrt{1-(\rho_p(w))^2}}\exp\left(-\frac{x_1^2-2\rho_p(w) x_1 x_2+x_2^2}{2(1-(\rho_p(w))^2)}\right)\frac{\partial}{\partial w_j}\left(-\frac{x_1^2-2\rho_p(w) x_1 x_2+x_2^2}{2(1-(\rho_p(w))^2)}\right).
\end{align*}
Since we have
\begin{align*}
&\frac{\partial}{\partial w_j}\left(-\frac{x_1^2-2\rho_p(w) x_1 x_2+x_2^2}{2(1-(\rho_p(w))^2)}\right)\\
&=\frac{-\left(\left(-2x_1x_2\right)\left(2(1-(\rho_p(w))^2)\right)-(x_1^2-2\rho_p(w)x_1x_2+x_2^2) 2(-2\rho_p(w))\right)}{(2(1-(\rho_p(w))^2))^2}
\left(\frac{\partial}{\partial w_j}\rho_p(w)\right)\\
&=\frac{-\left(-x_1x_2+x_1x_2(\rho_p(w))^2
+\rho_p(w)x_1^2-2(\rho_p(w))^2x_1x_2+\rho_p(w)x_2^2
\right)}{(1-(\rho_p(w))^2)^2}\left(\frac{\partial}{\partial w_j}\rho_p(w)\right)\\
&=-\frac{\rho_p(w)x_1^2-(1+(\rho_p(w))^2)x_1x_2+\rho_p(w)x_2^2}{(1-(\rho_p(w))^2)^2}\left(\frac{\partial}{\partial w_j}\rho_p(w)\right)
\end{align*}
and
\begin{align*}
&\frac{\partial }{\partial w_j}f^{\Phi_2}_{\rho_p(w)} (x_1, x_2)\\
&=\frac{1}{2\pi}\exp\left(-\frac{x_1^2-2\rho_p(w) x_1 x_2+x_2^2}{2(1-(\rho_p(w))^2)}\right)\left(1-(\rho_p(w))^2\right)^{-\frac{3}{2}}\rho_p(w) \frac{\partial}{\partial w_j}\rho_p(w)\\
&-\frac{\rho_p(w)x_1^2-(1+(\rho_p(w))^2)x_1x_2+\rho_p(w)x_2^2}{2\pi\sqrt{1-(\rho_p(w))^2}(1-(\rho_p(w))^2)^2}\exp\left(-\frac{x_1^2-2\rho_p(w) x_1 x_2+x_2^2}{2(1-(\rho_p(w))^2)}\right)\frac{\partial}{\partial w_j}\rho_p(w)
\\
&=\left(\frac{\rho_p(w)}{1-(\rho_p(w))^2}-\frac{\rho_p(w)x_1^2-(1+(\rho_p(w))^2)x_1x_2+\rho_p(w)x_2^2}{(1-(\rho_p(w))^2)^2}\right)f^{\Phi_2}_{\rho_p(w)} (x_1, x_2)\frac{\partial}{\partial w_j}\rho_p(w)
\end{align*}
which is \eqref{eq:d f 2-dim normal dy}.
\end{proof}

\begin{proof}[Proof of Proposition \ref{prop:mctcovar}]
Since $x=-\CoVaR{\eta}{\Xi_p(w)}{\zeta}{\Xi_0}$, we get
\[
\frac{\partial }{\partial w_j}\CoVaR{\eta}{\Xi_p(w)}{\zeta}{\Xi_0}=-\frac{\partial x}{\partial w_j}.
\]
By applying implicit differentiation, we have
\[
\frac{\partial x}{\partial w_j} = -\frac{\frac{\partial }{\partial w_j}G(x,w)}{\frac{\partial }{\partial x}G(x,w)},
\]
and hence 
\[
\frac{\partial }{\partial w_j}\CoVaR{\eta}{\Xi_p(w)}{\zeta}{\Xi_0}=\frac{\frac{\partial }{\partial w_j}G(x,w)}{\frac{\partial }{\partial x}G(x,w)} \Bigg|_{x=-\CoVaR{\eta}{\Xi_p(w)}{\zeta}{\Xi_0}}.
\]
We have
\begin{align*}
\frac{\partial }{\partial w_j} G(x,w) 
&
= \int_0^\infty \Bigg(\int_{-\infty}^{u(x,w,t)} 
\int_{-\infty}^{v(t)} 
\frac{\partial}{\partial w_j}f^{\Phi_2}_{\rho_p(w)}(x_1, x_2) dx_1 dx_2
\\
&~~
+ \int_{-\infty}^{v(t)}
f^{\Phi_2}_{\rho_p(w)}(x_1, u(x,w,t)) dx_1  
\frac{\partial}{\partial w_j}u(x,w,t) \Bigg)f_\Tau(t)dt
\end{align*}

By \eqref{eq:2dimNormal to 1dimNormal}, we have
\begin{align*}
&\frac{\partial }{\partial w_j} G(x,w) \\
&= \int_0^\infty \Bigg( 
\int_{-\infty}^{u(x,w,t)} 
\int_{-\infty}^{v(t)} 
\frac{\partial}{\partial w_j}f^{\Phi_2}_{\rho_p(w)}(x_1, x_2) dx_1 dx_2
\\
&~~
+\frac{1}{\sqrt{2\pi}}\exp\left(-\frac{\left(u(x,w,t)\right)^2}{2}\right) 
F_{\Phi}\left(%
	\frac{
		v(t)
		 -\rho_p(w)u(x,w,t)}
	 {\sqrt{1-(\rho_p(w))^2}}%
\right)  
\frac{\partial}{\partial w_j}u(x,w,t) \Bigg) f_\Tau(t) dt
\end{align*}

Hence, by \eqref{eq:d f 2-dim normal dy}, we have
\begin{align*}
&\int_0^\infty \int_{-\infty}^{u(x,w,t)} \int_{-\infty}^{v(t)} 
\frac{\partial}{\partial w_j}f^{\Phi_2}_{\rho_p(w)}(x_1, x_2) dx_1 dx_2 f_\Tau(t) dt
\\
&
=\left(\frac{\partial}{\partial w_j}\rho_p(w)\right)\int_0^\infty \Bigg( \int_{-\infty}^{u(x,w,t)} \int_{-\infty}^{v(t)} 
\frac{\rho_p(w)}{1-(\rho_p(w))^2} f^{\Phi_2}_{\rho_p(w)} (x_1, x_2) dx_1 dx_2 \\
&- \int_{-\infty}^{u(x,w,t)} \int_{-\infty}^{v(t)}\frac{\rho_p(w)x_1^2-(1+(\rho_p(w))^2)x_1x_2+\rho_p(w)x_2^2}{(1-(\rho_p(w))^2)^2}
f^{\Phi_2}_{\rho_p(w)} (x_1, x_2) dx_1 dx_2 
\Bigg) f_\Tau(t) dt 
\\
&
=\left(\frac{\partial}{\partial w_j}\rho_p(w)\right)
\Bigg( \frac{\rho_p(w)}{1-(\rho_p(w))^2}F_{(\Xi_0, \Xi_p(w))}(v_0, x)\\
&- \int_0^\infty \int_{-\infty}^{u(x,w,t)} \int_{-\infty}^{v(t)}\frac{\rho_p(w)x_1^2-(1+(\rho_p(w))^2)x_1x_2+\rho_p(w)x_2^2}{(1-(\rho_p(w))^2)^2}
f^{\Phi_2}_{\rho_p(w)} (x_1, x_2) dx_1 dx_2 
 f_\Tau(t) dt \Bigg)
\end{align*}
Assume that  $(\epsilon_0, \epsilon_p)$ is the bivariate standard normal random vector with covariance $\rho_p(w)$ and
$\Tau$  is the tempered stable subordinator with parameters $(\alpha, \theta)$ independent of $(\epsilon_0, \epsilon_p)$.
Then we have
\begin{align*}
&\int_0^\infty \int_{-\infty}^{u(x,w,t)} \int_{-\infty}^{v(t)} 
\frac{\partial}{\partial w_j}f^{\Phi_2}_{\rho_p(w)}(x_1, x_2) dx_1 dx_2 f_\Tau(t) dt
 \\
 &
 =\frac{\rho_p(w)}{1-(\rho_p(w))^2} \left(\frac{\partial}{\partial w_j}\rho_p(w)\right)F_{(\Xi_0, \Xi_p(w))}\left(v_0, x\right)\\
&- \frac{\frac{\partial}{\partial w_j}\rho_p(w)}{(1-(\rho_p(w))^2)^2} E\left[(\rho_p(w)\epsilon_p^2-(1+(\rho_p(w))^2)\epsilon_p\epsilon_0+\rho_p(w)\epsilon_0^2)1_{\epsilon_p<u(x,w,\Tau)}1_{\epsilon_0<v(\Tau)}\right],
\end{align*}
and hence
\begin{align*}
&\frac{\partial }{\partial w_j} G(x,w) \\
 &
 =\frac{\rho_p(w)}{1-(\rho_p(w))^2}\left( \frac{\partial}{\partial w_j}\rho_p(w)\right)F_{(\Xi_0, \Xi_p(w))}\left(v_0, x\right)\\
&- \frac{\frac{\partial}{\partial w_j}\rho_p(w)}{(1-(\rho_p(w))^2)^2} E\left[(\rho_p(w)\epsilon_p^2-(1+(\rho_p(w))^2)\epsilon_p\epsilon_0+\rho_p(w)\epsilon_0^2)1_{\epsilon_p<u(x,w,\Tau)}1_{\epsilon_0<v(\Tau)}\right] 
\\
&~~
+\frac{1}{\sqrt{2\pi}}E\left[\exp\left(-\frac{\left(u(x,w,\Tau)\right)^2}{2}\right) 
F_{\Phi}\left(%
	\frac{
		v(\Tau)
		 -\rho_p(w)u(x,w,\Tau)}
	 {\sqrt{1-(\rho_p(w))^2}}%
\right)  
\frac{\partial}{\partial w_j}u(x,w,\Tau) \right],
\end{align*}
which is \eqref{eq:dGdwj}. On the other hand, we have
\begin{align*}
&\frac{\partial}{\partial x} G(x,w)\\
&=
\int_0^\infty\int_{-\infty}^{v(t)}
f^{\Phi_2}_{\rho_p(w)}(x_1, u(x,w,t))dx_1 
\frac{\partial}{\partial x}u(x,w,t) f_\Tau(t) dt
\\
&
=
\int_0^\infty
\frac{1}{\sqrt{2\pi}}\exp\left(-\frac{\left(u(x,w,t)\right)^2}{2}\right) 
F_{\Phi}\left(%
	\frac{
		v(t)
		 -\rho_p(w)u(x,w,t)}
	 {\sqrt{1-(\rho_p(w))^2}}%
\right)  
\frac{\partial}{\partial x}u(x,w,t) f_\Tau(t) dt
\end{align*}
by \eqref{eq:2dimNormal to 1dimNormal}. Since we have
\[
\frac{\partial}{\partial x}u(x,w,t)=\frac{1}{\gamma_p(w)\sqrt{t}},
\]
we obtain
\begin{align*}
&\frac{\partial}{\partial x} G(x,w)
\\
&
=
\int_0^\infty
\frac{1}{\sqrt{2\pi}\gamma_p(w)\sqrt{t}}\exp\left(-\frac{\left(u(x,w,t)\right)^2}{2}\right) 
F_{\Phi}\left(%
	\frac{
		v(t)
		 -\rho_p(w)u(x,w,t)}
	 {\sqrt{1-(\rho_p(w))^2}}%
\right)  
 f_\Tau(t) dt
\\
&
=
E\left[
\frac{1}{\sqrt{2\pi}\gamma_p(w)\sqrt{\Tau}}\exp\left(-\frac{\left(u(x,w,\Tau)\right)^2}{2}\right) 
F_{\Phi}\left(%
	\frac{
		v(\Tau)
		 -\rho_p(w)u(x,w,\Tau)}
	 {\sqrt{1-(\rho_p(w))^2}}%
\right)  \right]
\end{align*}
which is \eqref{eq:dGdx}.
\end{proof}

\begin{proof}[Proof of Proposition \ref{prop:mctCoCVaR}]
From \eqref{eq:CoCVaR Integral Form} and the definition of CoCVaR, We have
\begin{align*}
\CoCVaR{\eta}{\Xi_p(w)}{\zeta}{\Xi_0} = -\frac{H(w)}{F(w)},
\end{align*}
where
\begin{equation*}
H(w) = \int_0^\infty \int_{-\infty}^{C(w,t)} \int_{-\infty}^{v(t)}
   \left(\beta_p(w)(t-1)+x_2\gamma_p(w)\sqrt{t}\right)f^{\Phi_2}_{\rho_p(w)}(x_1, x_2) dx_1\,dx_2\,f_\Tau(t)dt
\end{equation*}
and
\begin{equation*}
F(w) = F_{(\Xi_0, \Xi_p(w))}\left(v_0, -\CoVaR{\eta}{\Xi_p(w)}{\zeta}{\Xi_0} \right).
\end{equation*}
Thus, we have
\begin{equation}\label{eq:dCoVaR dwj 0}
\frac{\partial}{\partial w_j}\CoCVaR{\eta}{\Xi_p(w)}{\zeta}{\Xi_0} = -\frac{1}{F(w)}\left(\frac{\partial}{\partial w_j}H(w)+\CoCVaR{\eta}{\Xi_p(w)}{\zeta}{\Xi_0}\frac{\partial}{\partial w_j}F(w)\right).
\end{equation}

\noindent\textit{(i) Calculating $\frac{\partial}{\partial w_j}H(w)$}\\
We have
\begin{align*}
&\frac{\partial}{\partial w_j}H(w)\\
&
  =\int_0^\infty \Bigg( \int_{-\infty}^{v(t)}
   \left(\beta_p(w)(t-1)+C(w,t)\gamma_p(w)\sqrt{t}\right) 
 f^{\Phi_2}_{\rho_p(w)}(x_1, C(w,t)) dx_1 \frac{\partial}{\partial w_j} C(w,t)
\\
& 
+ \int_{-\infty}^{C(w,t)} \int_{-\infty}^{v(t)}
  \frac{\partial}{\partial w_j}\left( \left(\beta_p(w)(t-1)+x_2\gamma_p(w)\sqrt{t}\right)f^{\Phi_2}_{\rho_p(w)}(x_1, x_2) \right) dx_1\,dx_2 \Bigg)f_\Tau(t)dt
\end{align*}
Let 
\[
I_1=\int_0^\infty \int_{-\infty}^{v(t)}
   \left(\beta_p(w)(t-1)+C(w,t)\gamma_p(w)\sqrt{t}\right) 
 f^{\Phi_2}_{\rho_p(w)}(x_1, C(w,t)) dx_1
 \frac{\partial}{\partial w_j} C(w,t)
 f_\Tau(t)dt
\]
and 
\[
I_2=\int_0^\infty\int_{-\infty}^{C(w,t)} \int_{-\infty}^{v(t)}
  \frac{\partial}{\partial w_j}\left( \left(\beta_p(w)(t-1)+x_2\gamma_p(w)\sqrt{t}\right)f^{\Phi_2}_{\rho_p(w)}(x_1, x_2) \right) dx_1\,dx_2 f_\Tau(t)dt.
\]
Then $\frac{\partial}{\partial w_j}H(w)=I_1+I_2$. 

Let's simplify $I_1$.
Note that $\beta_p(w)(t-1)+C(w,t)\gamma_p(w)\sqrt{t}=-\CoVaR{\eta}{\Xi_p(w)}{\zeta}{\Xi_0}$ by the definition of $C(w,t)$.
By \eqref{eq:2dimNormal to 1dimNormal}, we have
\begin{align*}
I_1
 &=
 \int_0^\infty
 \frac{-\CoVaR{\eta}{\Xi_p(w)}{\zeta}{\Xi_0}}{\sqrt{2\pi}}
  \exp\left( -\frac{\left(C(w,t)\right)^2}{2} \right) 
	F_{\Phi}\left(\frac{v(t)-\rho_p(w)C(w,t)} {\sqrt{1-(\rho_p(w))^2}}\right)  
\frac{\partial}{\partial w_j}C(w,\Tau) f_\Tau(t) dt
\\
&= \frac{-\CoVaR{\eta}{\Xi_p(w)}{\zeta}{\Xi_0}}{\sqrt{2\pi}} E\left[\exp\left( -\frac{\left(C(w,\Tau)\right)^2}{2} \right) 
	F_{\Phi}\left(\frac{v(\Tau)-\rho_p(w)C(w,\Tau)} {\sqrt{1-(\rho_p(w))^2}}\right)  
\frac{\partial}{\partial w_j}C(w,\Tau) \right],
\end{align*}
where $\Tau$ is the tempered stable subordinator with parameters $(\alpha, \theta)$.

Consider the integral $I_2$. We have
\begin{align*}
&
\frac{\partial}{\partial w_j}\left( \left(\beta_p(w)(t-1)+x_2\gamma_p(w)\sqrt{t}\right)f^{\Phi_2}_{\rho_p(w)}(x_1, x_2) \right)
\\
&
=f^{\Phi_2}_{\rho_p(w)}(x_1, x_2)\frac{\partial}{\partial w_j} \left(\beta_p(w)(t-1)+x_2\gamma_p(w)\sqrt{t}\right)
\\
&
+
\left(\beta_p(w)(t-1)+x_2\gamma_p(w)\sqrt{t}\right)\frac{\partial}{\partial w_j}f^{\Phi_2}_{\rho_p(w)}(x_1, x_2) 
\\
&
=f^{\Phi_2}_{\rho_p(w)}(x_1, x_2)\left((t-1)\frac{\partial}{\partial w_j} \beta_p(w)+x_2\sqrt{t}\frac{\partial}{\partial w_j} \gamma_p(w)\right)
\\
&
+
\left(\beta_p(w)(t-1)+x_2\gamma_p(w)\sqrt{t}\right)\frac{\partial}{\partial w_j}f^{\Phi_2}_{\rho_p(w)}(x_1, x_2).
\end{align*}
By \eqref{eq:d gamma d wj} and  \eqref{eq:d f 2-dim normal dy}, we have
\begin{align*}
&
\frac{\partial}{\partial w_j}\left( \left(\beta_p(w)(t-1)+x_2\gamma_p(w)\sqrt{t}\right)f^{\Phi_2}_{\rho_p(w)}(x_1, x_2) \right)\\
&
=f^{\Phi_2}_{\rho_p(w)}(x_1, x_2)\left((t-1)-\frac{x_2\beta_p(w)\left(\frac{2\theta}{2-\alpha}\right)\sqrt{t}}{\gamma_p(w)}\right) \frac{\partial}{\partial w_j} \beta_p(w)
\\
&
+\left(\beta_p(w)(t-1)+x_2\gamma_p(w)\sqrt{t}\right)\\
&\times
 \left(\frac{\rho_p(w)}{1-(\rho_p(w))^2}
 -\frac{\rho_p(w)x_1^2-(1+(\rho_p(w))^2)x_1x_2+\rho_p(w)x_2^2}{(1-(\rho_p(w))^2)^2}\right)
 \\
 &\times
 f^{\Phi_2}_{\rho_p(w)}(x_1, x_2)\frac{\partial}{\partial w_j}\rho_p(w).
\end{align*}
Hence, we obtain
\begin{align*}
I_2&=  \int_0^\infty \int_{-\infty}^{C(w,t)} \int_{-\infty}^{v(t)}
\left((t-1)-\frac{x_2\beta_p(w)\left(\frac{2\theta}{2-\alpha}\right)\sqrt{t}}{\gamma_p(w)}\right)f^{\Phi_2}_{\rho_p(w)}(x_1, x_2) \frac{\partial}{\partial w_j} \beta_p(w)
\\
&
+\left(\beta_p(w)(t-1)+x_2\gamma_p(w)\sqrt{t}\right)
\\
&\times
 \left(\frac{\rho_p(w)}{1-(\rho_p(w))^2}-\frac{\rho_p(w)x_1^2-(1+(\rho_p(w))^2)x_1x_2+\rho_p(w)x_2^2}{(1-(\rho_p(w))^2)^2}\right)\frac{\partial}{\partial w_j}\rho_p(w)
\\
&\times
 f^{\Phi_2}_{\rho_p(w)}(x_1, x_2)
 dx_1\,dx_2\,f_\Tau(t)dt. 
\end{align*}
Let
\[
J_1=\int_0^\infty \int_{-\infty}^{C(w,t)} \int_{-\infty}^{v(t)}
\left((t-1)-\frac{x_2\beta_p(w)\left(\frac{2\theta}{2-\alpha}\right)\sqrt{t}}{\gamma_p(w)}\right)f^{\Phi_2}_{\rho_p(w)}(x_1, x_2)f_\Tau(t)dt,
\]
\[
J_2 = \int_0^\infty \int_{-\infty}^{C(w,t)} \int_{-\infty}^{v(t)} \left(\beta_p(w)(t-1)+x_2\gamma_p(w)\sqrt{t}\right) f^{\Phi_2}_{\rho_p(w)}(x_1, x_2)
 dx_1\,dx_2\,f_\Tau(t)dt,
\]
and
\begin{align*}
J_3
&=\int_0^\infty \int_{-\infty}^{C(w,t)} \int_{-\infty}^{v(t)}
\left(\beta_p(w)(t-1)+x_2\gamma_p(w)\sqrt{t}\right)
\\
&
\times
\left(
\frac{\rho_p(w)x_1^2-(1+(\rho_p(w))^2)x_1x_2+\rho_p(w)x_2^2}{(1-(\rho_p(w))^2)^2}
\right)
f^{\Phi_2}_{\rho_p(w)}(x_1, x_2)
 dx_1\,dx_2\,f_\Tau(t)dt
\end{align*}
Then we have 
\[
I_2 = J_1 \frac{\partial}{\partial w_j} \beta_p(w) +\left(\frac{\rho_p(w)}{1-(\rho_p(w))^2}J_2-J_3\right)\frac{\partial}{\partial w_j}\rho_p(w).
\]
We get $J_2=-F(w)\CoCVaR{\eta}{\Xi_p(w)}{\zeta}{\Xi_0}$ by \eqref{eq:CoCVaR Integral Form}.
Assume $(\epsilon_0, \epsilon_p)$ is the bivariate standard normal distributed random vector with covariance $\rho_p(w)$, and $\Tau$ is the tempered stable subordinator with parameters $(\alpha, \theta)$ independent of $(\epsilon_0, \epsilon_p)$.
Then 
\[
J_1 =  E\left[\left((\Tau-1)-\frac{\epsilon_p\beta_p(w)\left(\frac{2\theta}{2-\alpha}\right)\sqrt{\Tau}}{\gamma_p(w)}\right) 1_{\epsilon_p<C(w,\Tau)}1_{\epsilon_0<v(\Tau)}\right]
\]
and
\[
J_3 = E\left[ \left(\frac{\rho_p(w)\epsilon_0^2-(1+(\rho_p(w))^2)\epsilon_0 \epsilon_p+\rho_p(w) \epsilon_p^2}{(1-(\rho_p(w))^2)^2}\right)\xi_p(w) 1_{\epsilon_p<C(w,\Tau)}1_{\epsilon_0<v(\Tau)}\right]
\]
where $\xi_p(w) = \beta_p(w)(\Tau-1)+\epsilon_p \gamma_p(w)\sqrt{\Tau}$. Substituting $I_1$, $J_1$, $J_2$, $J_3$ into 
\[
\frac{\partial}{\partial w_j}H(w)
=I_1+J_1 \frac{\partial}{\partial w_j} \beta_p(w) +\left(\frac{\rho_p(w)}{1-(\rho_p(w))^2}J_2-J_3\right)\frac{\partial}{\partial w_j}\rho_p(w),
\]
we obtain 
\begin{align}
\label{eq:dH(w)dwj}
&\frac{\partial}{\partial w_j}H(w)
\\
\nonumber
&
=
\frac{-\CoVaR{\eta}{\Xi_p(w)}{\zeta}{\Xi_0}}{\sqrt{2\pi}} E\left[\exp\left( -\frac{\left(C(w,\Tau)\right)^2}{2} \right) 
	F_{\Phi}\left(\frac{v(\Tau)-\rho_p(w)C(w,\Tau)} {\sqrt{1-(\rho_p(w))^2}}\right)  
\frac{\partial}{\partial w_j}C(w,\Tau) \right]
\\
\nonumber
&
+E\left[\left((\Tau-1)-\frac{\epsilon_p\beta_p(w)\left(\frac{2\theta}{2-\alpha}\right)\sqrt{\Tau}}{\gamma_p(w)}\right) 1_{\epsilon_p<C(w,\Tau)}1_{\epsilon_0<v(\Tau)}\right]\frac{\partial}{\partial w_j} \beta_p(w)
 \\
\nonumber
&
 - \frac{\rho_p(w)}{1-(\rho_p(w))^2} \left(\frac{\partial}{\partial w_j} \rho_p(w)\right)F(w) \CoCVaR{\eta}{\Xi_p(w)}{\zeta}{\Xi_0}
  \\
\nonumber
 &
 - E\left[ \left(\frac{\rho_p(w)\epsilon_0^2-(1+(\rho_p(w))^2)\epsilon_0 \epsilon_p+\rho_p(w) \epsilon_p^2}{(1-(\rho_p(w))^2)^2}\right)\xi_p(w) 1_{\epsilon_p<C(w,\Tau)}1_{\epsilon_0<v(\Tau)}\right]\frac{\partial}{\partial w_j} \rho_p(w).
\end{align}

(ii) Calculating $\frac{\partial}{\partial w_j}F(w)$\\
Since we have
\begin{align*}
F(w)=\int_0^\infty \int_{-\infty}^{C(w,t)} \int_{-\infty}^{v(t)} f^{\Phi_2}_{\rho_p(w)}(x_1, x_2) dx_1 dx_2 f_\Tau(t)dt,
\end{align*}
we get
\begin{align*}
\frac{\partial}{\partial w_j}F(w)
&
= \int_0^\infty \Bigg(\int_{-\infty}^{C(w,t)} 
\int_{-\infty}^{v(t)} 
\frac{\partial}{\partial w_j}f^{\Phi_2}_{\rho_p(w)}(x_1, x_2) dx_1 dx_2
\\
&~~
+ \int_{-\infty}^{v(t)}
f^{\Phi_2}_{\rho_p(w)}(x_1, u(x,w,t)) dx_1  
\frac{\partial}{\partial w_j}C(w,t) \Bigg)f_\Tau(t)dt
\end{align*}
Using the same arguments in the proof of Proposition \ref{prop:mctcovar}, we obtain
\begin{align}
\label{eq:dF(w)dwj}
&\frac{\partial }{\partial w_j} F(w) 
\\
\nonumber
&
 =\frac{\rho_p(w)}{1-(\rho_p(w))^2}\left( \frac{\partial}{\partial w_j}\rho_p(w)\right) F(w) 
 \\
 \nonumber
&
- E\left[\left(\frac{ \rho_p(w)\epsilon_p^2-(1+(\rho_p(w))^2)\epsilon_p\epsilon_0+\rho_p(w)\epsilon_0^2}{(1-(\rho_p(w))^2)^2}\right)1_{\epsilon_p<C(w,\Tau)}1_{\epsilon_0<v(\Tau)}\right] \frac{\partial}{\partial w_j}\rho_p(w)
\\
\nonumber
&~~
+\frac{1}{\sqrt{2\pi}}E\left[\exp\left(-\frac{\left(C(w,\Tau)\right)^2}{2}\right) 
F_{\Phi}\left(%
	\frac{
		v(\Tau)
		 -\rho_p(w)u(x,w,\Tau)}
	 {\sqrt{1-(\rho_p(w))^2}}%
\right)  
\frac{\partial}{\partial w_j}C(w,\Tau) \right].
\end{align}
By substituting \eqref{eq:dH(w)dwj}, \eqref{eq:dF(w)dwj}, and $F(w)=F_{(\Xi_0, \Xi_p(w))}\left(v_0, -\CoVaR{\eta}{\Xi_p(w)}{\zeta}{\Xi_0} \right) = \eta\zeta$ into \eqref{eq:dCoVaR dwj 0}, we get
\begin{align*}
&\frac{\partial}{\partial w_j}\CoCVaR{\eta}{\Xi_p(w)}{\zeta}{\Xi_0} 
\\
=&
 -\frac{1}{\eta\zeta}\Bigg(
E\left[\left((\Tau-1)-\frac{\epsilon_p\beta_p(w)\left(\frac{2\theta}{2-\alpha}\right)\sqrt{\Tau}}{\gamma_p(w)}\right) 1_{\epsilon_p<C(w,\Tau)}1_{\epsilon_0<v(\Tau)}\right]\frac{\partial}{\partial w_j} \beta_p(w)
\\
&
- E\left[ \left(\frac{\rho_p(w)\epsilon_0^2-(1+(\rho_p(w))^2)\epsilon_0 \epsilon_p+\rho_p(w) \epsilon_p^2}{(1-(\rho_p(w))^2)^2}\right)\xi_p(w) 1_{\epsilon_p<C(w,\Tau)}1_{\epsilon_0<v(\Tau)}\right]\frac{\partial}{\partial w_j} \rho_p(w)
\\
&
- \CoCVaR{\eta}{\Xi_p(w)}{\zeta}{\Xi_0} E\left[\left(\frac{\rho_p(w)\epsilon_p^2-(1+(\rho_p(w))^2)\epsilon_p\epsilon_0+\rho_p(w)\epsilon_0^2}{(1-(\rho_p(w))^2)^2}\right)1_{\epsilon_p<C(w,\Tau)}1_{\epsilon_0<v(\Tau)}\right] \frac{\partial}{\partial w_j}\rho_p(w)
\Bigg)
\end{align*}
Hence we obtain \eqref{eq:mctcocvar}.

\end{proof}


\singlespacing
\bibliographystyle{decsci_mod}
\bibliography{refs_aaron_mctcocvar}

\end{document}